\newcommand{\lv}[1]{#1}   
\newcommand{\sv}[1]{}

\documentclass[svgnames,11pt,letterpaper]{article}
\usepackage{todonotes}
\usepackage{amssymb,amsmath,amsthm}
\usepackage{vmargin,comment,tikz}
\usepackage{enumitem}
\usepackage{microtype}
\setpapersize{USletter}
\usetikzlibrary{shapes,arrows,fit,calc,positioning}

\setmarginsrb{1in}{1in}{1in}{1in}{0pt}{0pt}{0pt}{7mm}
\usepackage{times}

\numberwithin{subcase}{case}

\newtheorem{theorem}{Theorem}
\newtheorem{observation}{Observation}

\newtheorem{lemma}{Lemma}
\newtheorem{claim}{Claim}
\newtheorem{tlemma}{Lemma}[lemma]

\newtheorem{definition}{Definition}
\newtheorem{tdefinition}{Definition}[definition]

\usepackage{boxedminipage}
\usepackage{comment}

\newcommand{\Yes}{{\sc Yes}}
\newcommand{\No}{{\sc No}}
\newcommand{\FPT}{\text{\normalfont FPT}}
\newcommand{\W}[1][xxxx]{\text{\normalfont W[#1]}}

\newcommand{\bdcomp}{{\sc SBD Comp}}
\newcommand{\ebdcomp}{{{\sc Ext-}\bdcomp}}
\newcommand{\scite}{\cite}

\renewcommand{\uplus}{\oplus}

\newcommand{\splitclass}{$\CSP(\Gamma_1)\uplus\cdots\uplus\CSP(\Gamma_d)$}

\usepackage{amsthm}

\newcommand{\bigoh}{\mathcal{O}}

\newcommand{\NN}{\mathbb{N}}
\newcommand{\cB}{\mathcal{B}}

\newcommand{\C}{\mathbf{C}}

\newcommand{\cY}{\mathcal{Y}}

\newcommand{\cR}{\mathcal{R}}

\newcommand{\cD}{\mathcal{D}}

\newcommand{\cH}{\mathcal{H}}

\newcommand{\cP}{\mathcal{P}}

\newcommand{\cX}{\mathcal{X}}

\newcommand{\bC}{\mathbf{C}}

\newcommand{\sharpp}{$\text{\normalfont \#}$}

\newtheorem{reduction}{Preprocessing Rule}

  \usepackage[ pdftex, plainpages = false, pdfpagelabels, 
                 bookmarks=false,
                 bookmarksopen = true,
                 bookmarksnumbered = true,
                 breaklinks = true,
                 linktocpage,
                 pagebackref,
                 colorlinks = true,  
                 linkcolor = blue,
                 urlcolor  = blue,
                 citecolor = red,
                 anchorcolor = green,
                 hyperindex = true,
                 hyperfigures
                 ]{hyperref} 

\sv{\title{Discovering Archipelagos of Tractability for\\ Constraint
  Satisfaction and Counting\\[6pt]
\small (Extended Abstract)}}

\lv{

  \title{Discovering Archipelagos of Tractability for\\ Constraint
  Satisfaction and Counting
  \thanks{Research supported by the Austrian
    Science Funds (FWF), project P26696 X-TRACT.}}}

\author{Robert Ganian\\
\small Vienna University of Technology\\[-2pt] 
\small Vienna, Austria\\
\small  rganian@gmail.com
\and 
M.S. Ramanujan\\
\small University of Bergen\\[-2pt]
\small Bergen, Norway\\
\small  ramanujan.sridharan@ii.uib.no
\and 
Stefan Szeider\\
\small Vienna University of Technology\\[-2pt] 
\small Vienna, Austria\\
\small  stefan@szeider.net
}
\date{}

\newcommand{\sharpCSP}{\text{\normalfont \#CSP}}

\begin{document}
\maketitle

\thispagestyle{empty} 
\begin{abstract}
  The Constraint Satisfaction Problem (CSP) is a central and generic
  computational problem which provides a common framework for many
  theoretical and practical applications. A central line of research
  is concerned with the identification of classes of instances for
  which CSP can be solved in polynomial time; such classes are
  often called ``islands of tractability.''  A prominent way of
  defining islands of tractability for CSP is to restrict the
  relations that may occur in the constraints to a fixed set, called a
  \emph{constraint language}, whereas a constraint language is conservative if
  it contains all unary relations.  Schaefer's famous Dichotomy
  Theorem (STOC 1978) identifies all islands of tractability in terms
  of tractable constraint languages over a Boolean domain of values.
  Since then many extensions and generalizations of this result have
  been obtained.
  Recently, Bulatov (TOCL 2011, JACM 2013) gave a full
  characterization of all islands of tractability for CSP and the
  counting version {\sharpCSP} that are defined in terms of conservative
  constraint languages.  

  This paper addresses the general limit of the mentioned tractability
  results for CSP and \sharpCSP, that they only apply to instances
  where all constraints belong to a single tractable language (in
  general, the union of two tractable languages isn't tractable).  We
  show that we can overcome this limitation as long as we keep some
  control of how constraints over the various considered tractable
  languages interact with each other. For this purpose we utilize the
  notion of a \emph{strong backdoor} of a CSP instance, as introduced
  by Williams et al.~(IJCAI 2003), which is a set of variables that
  when instantiated moves the instance to an island of tractability,
  i.e., to a tractable class of instances. We consider strong
  backdoors into \emph{scattered classes}, consisting of CSP instances
  where each connected component belongs entirely to some class from a
  list of tractable classes. Figuratively speaking, a scattered class
  constitutes an \emph{archipelago of tractability}.  The main difficulty
  lies in finding a strong backdoor of given size $k$; once it is
  found, we can try all possible instantiations of the backdoor
  variables and apply the polynomial time algorithms associated with
  the islands of tractability on the list component wise.  Our main
  result is an algorithm that, given a CSP instance with $n$
  variables, finds in time $f(k)n^{\bigoh(1)}$ a strong backdoor into
  a scattered class (associated with a list of finite conservative
  constraint languages) of size $k$ or correctly decides that there
  isn't such a backdoor.
  This also gives the running time for solving (\#)CSP, provided that
  (\#)CSP is polynomial-time tractable for the considered
  constraint languages. Our result makes significant progress towards the main goal of the backdoor-based approach to CSPs -- the identification of maximal base classes for which small backdoors can be detected efficiently. 
\end{abstract}

\pagebreak
\section{Introduction}


\newcommand{\SB}{\{\,}%
\newcommand{\SM}{\;{:}\;}%
\newcommand{\SE}{\,\}}%

\newcommand{\var}{\text{\normalfont var}}
\newcommand{\spbd}{\text{\normalfont split-bd-size}}

\newcommand{\rel}{\text{\normalfont rel}}
\newcommand{\fun}{\text{\normalfont fun}}

\newcommand{\relD}{\DDD^*}
\newcommand{\CCC}{\mathcal{C}}
\newcommand{\HHH}{\mathcal{H}}
\newcommand{\III}{\mathbf{I}}
\newcommand{\KKK}{\mathbf{K}}
\newcommand{\VVV}{\mathcal{V}}
\newcommand{\DDD}{\mathcal{D}}
\newcommand{\FFF}{\mathcal{F}}
\newcommand{\CSP}{\text{\normalfont CSP}}
\newcommand{\VCSP}{\text{\normalfont VCSP}}
\newcommand{\strongbds}{\textsc{SBD}($\CSP(\Gamma_1)\uplus\cdots\uplus\CSP(\Gamma_d)$)}

The Constraint Satisfaction Problem (CSP) is a central and generic
computational problem which provides a common framework for
many theoretical and practical applications \cite{HellNesetril08}. An
instance of CSP consists of a collection of variables that must be
assigned values subject to constraints, where each constraint is given
in terms of a relation whose tuples specify the allowed combinations
of values for specified variables. The problem was originally
formulated by Montanari \cite{Montanari74}, and has been found
equivalent to the homomorphism problem for relational
structures \cite{FederVardi98} and the problem of evaluating
conjunctive queries on databases~\cite{Kolaitis03}.  In general
CSP is NP-complete. A central line of research is concerned with
the identification of classes of instances for which CSP can be
solved in polynomial time. Such classes are often called ``islands of
tractability'' \cite{Kolaitis03,KolaitisVardi07}.  
  
A prominent way of defining islands of tractability for CSP is to
restrict the relations that may occur in the constraints to a fixed
set $\Gamma$, called a \emph{constraint language}.  A finite
constraint language is \emph{tractable} if CSP restricted to instances
using only relations from $\Gamma$, denoted $\CSP(\Gamma)$, can be
solved in polynomial time. 
Schaefer's famous Dichotomy Theorem~\cite{Schaefer78}
identifies all islands of tractability in terms of tractable
constraint languages over the two-element domain. Since then, many
extensions and generalizations of this result have been obtained
\cite{JeavonsCohenGyssens97,Creignou95,KolmogorovZivny13,ThapperZivny13}.
The Dichotomy Conjecture of Feder and Vardi \cite{FederVardi93} claims
that for every finite constraint language~$\Gamma$, $\CSP(\Gamma)$ is
either NP-complete or solvable in polynomial time.  Schaefer's
Dichotomy Theorem shows that the conjecture holds for two-element
domains; more recently, Bulatov~\cite{Bulatov06} showed the conjecture
to be true for three-element domains.
Several papers are devoted to identifying constraint languages
$\Gamma$ for which \emph{counting CSP}, denoted
$\sharpCSP(\Gamma)$, can be solved in polynomial time
\cite{Bulatov13,CreignouHermann96,BulatovDalmau07}, i.e., where the
number of satisfying assignments can be computed in polynomial time.
Such languages $\Gamma$ are called \sharpp-\emph{tractable}.

A constraint language over $\DDD$ is \emph{conservative} if it
contains all possible unary constraints over $\DDD$, and it is
\emph{semi-conservative} if it contains all possible unary constant
constraints (i.e., constraints that fix a variable to a specific
domain element).  
 These properties of  constraint languages
 are very natural, as one would expect in practical settings that the
 unary relations are present.  Indeed, some authors (e.g.,
 \cite{CooperCohenJeavons94}) even define CSP so that every variable
 can have its own set of domain values, making conservativeness a
 built-in property. 
Recently, Bulatov \cite{Bulatov11} gave a full characterization of all
tractable conservative constraint languages over finite
domains. 
Furthermore, Bulatov~\cite{Bulatov13} gave a full characterization of
all \sharpp-tractable constraint languages over finite domains.  Thus,
Bulatov's results identify all islands of (\sharpp-)tractability over
finite domains which can be defined in terms of a conservative
constraint language.
 
A general limit of tractability results for CSP and {\sharpCSP} based on
constraint languages, such as the mentioned results of Schaefer and
Bulatov, is that they only apply to instances where all constraints
belong to a single tractable language. One cannot arbitrarily combine
constraints from two or more tractable languages, as in general, the
union of two tractable languages isn't tractable (see
Section~\ref{sect:prel}).  In this paper we show that we can
overcome this limitation as long as constraints over the various considered tractable languages interact
with each other in a controlled manner. For this purpose we utilize the notion of a
\emph{strong backdoor} of a CSP instance, as introduced by Williams et
al.~\cite{WilliamsGomesSelman03}. A set $B$ of variables of a CSP
instance is a strong backdoor  into a tractable class $\HHH$ if for
all instantiations of the variables in $B$, the reduced instance
belongs to~$\HHH$. In this paper, we consider strong backdoors into a
\emph{scattered class}, denoted $\HHH_1 \uplus \cdots \uplus
\HHH_d$, consisting of all CSP instances $\III$ such that each
connected component of $\III$ belongs entirely to some class from a
list of tractable classes $\HHH_1,\dots,\HHH_d$. Figuratively
speaking, $\HHH_1 \uplus \dots \uplus \HHH_d$ constitutes an
archipelago of tractability, consisting of the islands
$\HHH_1,\dots,\HHH_d$.

\pagebreak
Our main result  is the following:
\begin{theorem} \label{thm:main} Let $\Gamma_1,\dots,\Gamma_d$ be
  semi-conservative finite constraint languages over 
  domain~$\DDD$, and let $\relD$ be the language containing all
  relations over $\DDD$.  
  If $\Gamma_1,\dots,\Gamma_d$ are tractable
  (or \sharpp-tractable), then $\CSP(\relD)$ (or $\sharpCSP(\relD)$,
  respectively) can be solved in time $2^{2^{\bigoh(k)}} \cdot
  n^{\bigoh(1)}$ for instances with $n$ variables that have a strong
  backdoor of size $k$ into $\CSP(\Gamma_1) \uplus \dots \uplus
  \CSP(\Gamma_d)$.
\end{theorem}

 Note that there are natural CSP
instances which have a small strong backdoor into the scattered
class  $\CSP(\Gamma_1) \uplus \dots \uplus
  \CSP(\Gamma_d)$ but require strong backdoors of arbitrarily large size
into each individual base class $\CSP(\Gamma_i)$.
The power of a strong backdoor into a scattered class over one into a
single class stems from the fact that
the instantiation of variables in the
backdoor can serve two purposes. The first is to separate constraints into
components, each belonging entirely to some $\CSP(\Gamma_i)$ (possibly even different $\CSP(\Gamma_i)$'s for different instantiations), and the second is to modify constraints so that once modified, the component containing these constraints belongs to some $\CSP(\Gamma_i)$.
 
 

When using the backdoor-based approach, the main computational difficulty is in detecting small backdoor sets into the chosen base class. This task becomes significantly harder when the base classes are made more general.
However, we show that while scattered classes are significantly more general than single base classes, we can still detect strong backdoors into such classes in FPT time. The formal statement of this result, which represents our main technical contribution, is the following.

\begin{lemma}\label{lem:main-1} There is an  algorithm that, given a CSP instance $\III$ and a parameter
$k$, runs in time $2^{2^{\bigoh(k)}} \cdot
  n^{\bigoh(1)}$ and either finds a strong backdoor of size at most $k$ in $\III$ into
$\CSP(\Gamma_1^*) \uplus \dots \uplus \CSP(\Gamma_d^*)$ or correctly
decides that none exists. 
\end{lemma}
 Here $\Gamma_i^*\supseteq
\Gamma_i$ is obtained from $\Gamma_i$ by taking the closure under
partial assignments and by adding a redundant relation.

We remark that the finitary restriction on the constraint languages is
unavoidable, since otherwise the arity of the relations or the domain
size would be unbounded. However, for unbounded arity, small
backdoors cannot be found efficiently as Lemma~\ref{lem:main-1} would
not hold already for the special case of $d=1$ unless
$\FPT=\W[2]$~\cite{GaspersMisraOrdyniakSzeiderZivny14}.  Similarly,
with unbounded domain, a small strong backdoor cannot be used
efficiently. For instance, the natural encoding of the $\W[1]$-hard
$k$-clique problem to CSP~\cite{PapadimitriouYannakakis99} only has
$k$ variables, and therefore has a size-$k$ strong backdoor to any
base class that contains the trivial constrains with empty scopes,
which is the case for any natural base class; an FPT algorithm solving
such instances would once again imply $\FPT=\W[1]$.

The following is a brief summary of the algorithm of Lemma~\ref{lem:main-1}. 
We will give a more detailed summary in Section \ref{sect:csp}.
\begin{enumerate}
\item We begin by using the technique of iterative
  compression~\cite{ReedSV04} to transform the problem into a structured subproblem which we call \textsc{Extended
    {\strongbds} Compression} ({\ebdcomp}). In this technique, the
  idea is to start with a sub-instance and a trivial solution for this
  sub-instance and iteratively expand the sub-instances while
  compressing the solutions till we solve the problem on the original
  instance. Specifically, in {\ebdcomp} we are given additional information about the desired solution in the input: we receive an ``old'' strong backdoor which is slightly bigger than our target size, along with information about how this old backdoor interacts with our target solution. This is formalized in Subsection~\ref{sub:ic}.
\item In Subsection~\ref{sub:nonsep}, we consider only solutions for {\ebdcomp} instances
  which have a certain `inseparability property' and give an FPT
  algorithm to test for the presence of such solutions. To be more precise, here we only look for solutions of {\ebdcomp} which leave the omitted part of the old strong backdoor in a single connected component.
  We handle this case separately at the beginning since it serves
  as a base case in our algorithm to solve general instances. Interestingly, even this base case requires the extension of state of the art separator techniques to a CSP setting.
\item Finally, in Subsection~\ref{sub:general} we show how to handle general instances of
  {\ebdcomp}. This part of the algorithm relies on a new \emph{pattern replacement}
  technique, which shares certain superficial similarities with
  protrusion replacement~\cite{BodlaenderFLPST09} but allows the
  preservation of a much larger set of structural properties (such as
  containment of disconnected forbidden structures and connectivity
  across the boundary). We interleave our pattern replacement procedure with the recently developed approach of `important separator
sequences'~\cite{LokshtanovR12} as well as the algorithm designed in the previous subsection for `inseparable' instances in order to solve the problem on general instances. Before we conclude the summary, we would like to point out an interesting feature of our algorithm. At its very core, it is a branching algorithm; in FPT time we identify a bounded set of variables which intersects some solution and then branch on this set. Note that this approach does not always result in an FPT-algorithm for computing strong backdoor sets. In fact, depending on the base class it might only imply an FPT-\emph{approximation} algorithm (see \cite{GaspersSzeider13}). This is because we need to explore all possible assignments for the chosen variable. However, we develop a notion of \emph{forbidden sets of constraints} which allows us to succinctly describe when a particular set is not already a solution. Therefore, when we branch on a supposed strong backdoor variable, we simply add it to a partial solution which we maintain and then we can at any point easily check whether the partial solution is already a solution or not. This is a crucial component of our FPT algorithm.

\end{enumerate}

\paragraph{Related Work}
Williams et al.~\cite{WilliamsGomesSelman03,WilliamsGomesSelman03a}
introduced the notion of \emph{backdoors} for the runtime analysis of
algorithms for CSP and SAT, see also \cite{HemaspaandraWilliams12} for
a more recent discussion of backdoors for SAT.  A backdoor is a small
set of variables whose instantiation puts the instance into a fixed
tractable class. One distinguishes between strong and weak backdoors,
where for the former all instantiations lead to an instance in the
base class, and for the latter at least one leads to a satisfiable
instance in the base class.  Backdoors have been studied under a
different name by Crama et al.~\cite{CramaEkinHammer97}. The study of
the parameterized complexity of finding small backdoors was initiated
by Nishimura et al.~\cite{NishimuraRagdeSzeider04-informal} for SAT,
who considered backdoors into the classes of Horn and Krom CNF
formulas. Further results cover the classes of renamable Horn
formulas~\cite{RazgonOSullivan09}, q-Horn
formulas~\cite{GaspersOrdyniakRamanujanSaurabhSzeider13} and classes
of formulas of bounded
treewidth~\cite{GaspersSzeider13,RamanujanLokshtanovFominSaurabhMisra15}. The
detection of backdoors for CSP has been studied for instance
in~\cite{BessiereCarbonnelHebrardKatsirelosWalsh13,CarbonnelCooperHebrard14}. Gaspers
et al.~\cite{GaspersMisraOrdyniakSzeiderZivny14} recently obtained
results on the detection of strong backdoors into \emph{heterogeneous}
base classes of the form $\CSP(\Gamma_1)\cup\dots \cup \CSP(\Gamma_d)$
where for each instantiation of the backdoor variables, the reduced
instance belongs entirely to some $\CSP(\Gamma_i)$ (possibly to
different $\CSP(\Gamma_i)$'s for different instantiations).  Our
setting is more general since $\CSP(\Gamma_1)\uplus\cdots \uplus
\CSP(\Gamma_d) \supseteq \CSP(\Gamma_1)\cup\dots \cup \CSP(\Gamma_d)$,
and the size of a smallest strong backdoor into
$\CSP(\Gamma_1)\cup\dots \cup \CSP(\Gamma_d)$ can be arbitrarily
larger than the size of a smallest strong backdoor into
$\CSP(\Gamma_1)\uplus\cdots \uplus \CSP(\Gamma_d)$.

\sv{
\smallskip
\noindent\emph{Because of space restrictions some proofs have been shortened or
omitted. The full paper is appended.}
}
\section{Preliminaries}\label{sect:prel}

\lv{
\subsection{Constraint Satisfaction}
\label{sub:CSP}
}
\sv{ \paragraph{Constraint Satisfaction}}
Let $\VVV$ be an infinite set of variables and $\DDD$ a finite set of
values.  A \emph{constraint of arity $\rho$ over $\DDD$} is a pair $(S,R)$
where $S=(x_1,\dots,x_\rho)$ is a sequence of variables from $\VVV$ and
$R \subseteq \DDD^\rho$ is a $\rho$-ary relation. The set
$\var(C)=\{x_1,\dots,x_\rho\}$ is called the \emph{scope} of $C$.  A
\emph{value assignment} (or \emph{assignment}, for short) $\alpha:X\rightarrow \DDD$ is a mapping defined
on a set $X\subseteq \VVV$ of variables. An assignment
$\alpha:X\rightarrow \DDD$ \emph{satisfies} a constraint
$C=((x_1,\dots,x_\rho),R)$ if $\var(C)\subseteq X$ and
$(\alpha(x_1),\dots,\alpha(x_\rho)) \in R$.
For a set $\III$ of constraints we write $\var(\III)=\bigcup_{C\in \III}
\var(C)$ and $\rel(\III)=\SB R \SM (S,R) \in C, C\in \III \SE$.
\lv{

}
A finite set $\III$ of constraints is \emph{satisfiable} if there
exists an assignment that simultaneously satisfies all the constraints
in $\III$.  The \emph{Constraint Satisfaction Problem} (CSP, for
short) asks, given a finite set $\III$ of constraints, whether $\III$
is satisfiable.  The \emph{Counting Constraint Satisfaction Problem}
(\sharpCSP, for short) asks, given a finite set $\III$ of constraints, to
determine the number of assignments to $\var(\III)$ that satisfy
$\III$. {\CSP} is NP-complete and {\sharpCSP} is {\sharpp}P-complete
(see, e.g., \cite{Bulatov13}).
 \lv{
 
}
Let $\alpha:X\rightarrow \DDD$ be an assignment. For a $\rho$-ary
constraint $C=(S,R)$ with $S=(x_1,\dots,x_\rho)$ we denote by $C|_\alpha$
the constraint $(S',R')$ obtained from $C$ as follows. $R'$ is
obtained from $R$ by (i) deleting all tuples $(d_1,\dots,d_\rho)$ from
$R$ for which there is some $1\leq i \leq \rho$ such that $x_i\in X$ and $\alpha(x_i)\neq
d_i$, and (ii) removing from all remaining tuples all coordinates $d_i$
with $x_i\in X$.  $S'$ is obtained from $S$ by deleting all variables
$x_i$ with $x_i\in X$.  For a set $\III$ of constraints we define
$\III|_\alpha$ as $\SB C|_\alpha \SM C \in \III \SE$.  
\lv{

}
A \emph{constraint language} (or \emph{language}, for short) $\Gamma$
over a finite domain $\DDD$ is a set $\Gamma$ of relations (of
possibly various arities) over~$\DDD$.  By $\CSP(\Gamma)$ we denote
CSP restricted to instances $\III$ with $\rel(\III)\subseteq
\Gamma$.  A constraint language $\Gamma$ is \emph{tractable} if for every
finite subset $\Gamma'\subseteq \Gamma$, the problem $\CSP(\Gamma')$
can be solved in polynomial time.  A constraint language $\Gamma$ is
\emph{\sharpp-tractable} if for every finite subset $\Gamma'\subseteq
\Gamma$, the problem $\sharpCSP(\Gamma')$ can be solved in polynomial
time.

In his seminal paper \cite{Schaefer78}, Schaefer showed that for all
constraint languages $\Gamma$ over the Boolean domain $\{0,1\}$,
$\CSP(\Gamma)$ is either NP-complete or solvable in polynomial time.
In fact, he showed that a Boolean constraint language $\Gamma$ is
tractable if and only at least one of the following properties holds
for each relation $R\in \Gamma$: (i)~$(0,\dots,0)\in R$,
(ii)~$(1,\dots,1)\in R$, (iii)~$R$ is equivalent to a conjunction of
binary clauses, (iv)~$R$ is equivalent to a conjunction of Horn
clauses, (v)~$R$ is equivalent to a conjunction of dual-Horn clauses,
and (vi)~$R$ is equivalent to a conjunction of affine
formulas; $\Gamma$ is then called 1-valid, 0-valid, bijunctive, Horn,
dual-Horn, or affine, respectively.  A Boolean language that satisfies
any of these six properties is called a \emph{Schaefer language}.  A
constraint language $\Gamma$ over domain $\DDD$ is \emph{conservative}
if $\Gamma$ contains all unary relations over $\DDD$. Except for the
somewhat trivial 0-valid and 1-valid languages, all Schaefer languages
are conservative. $\Gamma$ is \emph{semi-conservative} if it contains
all unary relations over $\DDD$ that are singletons (i.e., constraints
that fix the value of a variable to some element of $\DDD$).

A constraint language $\Gamma$ is \emph{closed under assignments} if
for every $C=(S,R)$ such that $R\in \Gamma$ and every assignment $\alpha$,
it holds that $R'\in \Gamma$ where $C|_\alpha=(S', R')$. For a constraint
language $\Gamma$ over a domain~$\DDD$ we denote by $\Gamma^*$ the
smallest constraint language over $\DDD$ that contains $\Gamma\cup
\{\DDD^2\}$ and is closed under assignments; notice that $\Gamma^*$ is
uniquely determined by $\Gamma$. Evidently, if a language $\Gamma$ is
tractable (or \sharpp-tractable, respectively) and semi-conservative,
then so is~$\Gamma^*$: first, all constraints of the form $(S,\DDD^2|_\alpha)$ can be detected in polynomial time and removed from the
instance without changing the solution, and then each constraint
$C'=(S',R')$ with $R'\in \Gamma^*\setminus \Gamma$ can be expressed in
terms of the conjunction of a constraint $C=(S,R)$ with $R\in \Gamma$
and unary constraints over variables in $\var(C)\setminus \var(C')$.

As mentioned in the introduction, the union of two tractable constraint
languages is in general not tractable. Take for instance the conservative
languages $\Gamma_1=\{\{0,1\}^3\setminus \{(1,1,1)\}\}\cup
2^{\{0,1\}}$ and $\Gamma_2=\{\{0,1\}^3\setminus \{(0,0,0)\}\}\cup
2^{\{0,1\}}$.  Using the characterization of Schaefer languages in
terms of closure properties (see, e.g.,
\cite{GopalanKolaitisManevaPapadimitriou09}), it is easy to check that
$\Gamma_1$ is Horn and has none of the five other Schaefer properties;
similarly, $\Gamma_2$ is dual-Horn and has none of the five other Schaefer
properties.  Hence, if follows by Schaefer's Theorem that
$\CSP(\Gamma_1)$ and $\CSP(\Gamma_2)$ are tractable, but
$\CSP(\Gamma_1\cup \Gamma_2)$ is NP-complete. One can find similar
examples for other pairs of Schaefer languages.
\lv{
\subsection{Parameterized Complexity}
}
\sv{\vspace{-10 pt}\paragraph{Parameterized Complexity}}
A parameterized problem $\cP$ is a problem whose instances are tuples
$(I,k)$, where $k\in \NN$ is called the parameter. We say that a
parameterized problem is \emph{fixed parameter tractable} (FPT in
short) if it can be solved by an algorithm which runs in time
$f(k)\cdot |I|^{\bigoh(1)}$ for some computable function~$f$;
algorithms with running time of this form are called FPT algorithms.
The notions of \emph{$\W[i]$-hardness} (for $i\in \NN$) are frequently
used to show that a parameterized problem is not likely to be FPT; an
FPT algorithm for a $\W[i]$-hard problem would imply that the
Exponential Time Hypothesis
fails~\cite{ChenHuangKanjXia06}. 
We refer the reader to other sources
\cite{DowneyFellows99,DowneyFellows13,FlumGrohe06} for an in-depth
introduction into parameterized complexity.

\lv{
\subsection{Backdoors, Incidence Graphs and Scattered Classes}
\label{sub:BDIGCC}}
\sv{\vspace{-10 pt}
\paragraph{Backdoors, Incidence Graphs and Scattered Classes}}
Let $\III$ be an instance of CSP over $\DDD$ and let $\HHH$ be a class of CSP instances. 
A set $B$ of variables of $\III$ is called a \emph{strong backdoor} into $\HHH$ if for every assignment $\alpha: B\rightarrow \DDD$ it holds that $\III|_\alpha\in \HHH$. Notice that if we are given a strong backdoor $B$ of size $k$ into a tractable (or \sharpp-tractable) class $\HHH$, then it is possible to solve CSP (or \sharpCSP) in time $|\DDD|^k\cdot n^{\bigoh(1)}$. It is thus natural to ask for which tractable classes we can find a small backdoor efficiently. 

\begin{center}
\vspace{-0.5cm}
  \begin{boxedminipage}[t]{0.99\textwidth}
  \begin{quote}
  \textsc{Strong Backdoor Detection into $\HHH$ (SBD($\HHH$))}\\ \nopagebreak
  \emph{Setting}: A class $\HHH$ of CSP instances over a finite domain $\DDD$.\\ \nopagebreak
  \emph{Instance}: A CSP instance $\III$ over $\DDD$ and a non-negative integer $k$.\\ \nopagebreak
  \emph{Task}: Find a strong backdoor in $\III$ into $\HHH$ of cardinality at most $k$, or
   determine that no such strong backdoor exists.\\ \nopagebreak
  \emph{Parameter}: $k$.
\end{quote}
\end{boxedminipage}
\end{center}

We remark that for any finite constraint language $\Gamma$, the
problem \textsc{SBD}(CSP($\Gamma$)) is fixed parameter tractable due
to a simple folklore branching algorithm. On the other hand,
\textsc{SBD}(CSP($\Gamma'$)) is known to be $\W[2]$-hard for a wide
range of infinite tractable constraint languages
$\Gamma'$~\cite{GaspersMisraOrdyniakSzeiderZivny14}.
Given a CSP instance $\III$, we use ${\cB}(\III)=(\var(\III)\cup
\III,E)$ to denote the \emph{incidence graph} of $\III$; specifically,
$\III$ contains an edge $\{x,Y\}$ for $x\in\var(\III)$, $Y\in \III$ if
and only if $x\in \var(Y)$. We denote this graph by $\cB$ when $\III$
is clear from the context. Furthermore, for a set $S$ of variables of
$\III$, we denote by ${\cB}_S(\III)$ the graph obtained by deleting
from $\cB(\III)$ the vertices corresponding to the variables in $S$;
we may also use $\cB_S$ in short if $\III$ is clear from the
context. We refer to Diestel's book~\cite{Diestel12} for standard
graph terminology.

Two CSP instances $\III,$ $\III'$ are \emph{variable disjoint} if $\var(\III)\cap \var(\III')=\emptyset$.
Let $\HHH_1,\dots \HHH_d$ be classes of CSP instances. Then the \emph{scattered class} $\HHH_1\uplus \cdots \uplus \HHH_d$ is the class of all CSP instances $\III$ which may be partitioned into pairwise variable disjoint sub-instances $\III_1,\dots \III_d$ such that $\III_i\in \HHH_i$ for each $i\in [d]$. Notice that this implies that $\cB(\III)$ can be partitioned into pairwise disconnected subgraphs $\cB(\III_1),\dots \cB(\III_d)$.
If $\HHH_1,\dots \HHH_d$ are tractable, then $\HHH_1\uplus \cdots \uplus \HHH_d$ is also tractable, since each $\III_i$ can be solved independently. Similarly, If $\HHH_1,\dots \HHH_d$ are \sharpp-tractable, then $\HHH_1\uplus \cdots \uplus \HHH_d$ is also \sharpp-tractable, since the number of satisfying assignments in each $\III_i$ can be computed independently and then multiplied to obtain the solution.

We conclude this section by showcasing that a strong backdoor to a scattered class can be arbitrarily smaller than a strong backdoor to any of its component classes. Consider once again the tractable languages $\Gamma_1=\{\{0,1\}^3\setminus
\{(1,1,1)\}\}\cup 2^{\{0,1\}}$ (Horn) and $\Gamma_2=\{\{0,1\}^3\setminus
\{(0,0,0)\}\}\cup 2^{\{0,1\}}$ (dual-Horn). Then for any $k\in \mathbb{N}$ one can find $\III\in\CSP(\Gamma_1)\uplus\CSP(\Gamma_2)$ such that $\III$ does not have a strong backdoor of size $k$ to either of $\CSP(\Gamma_1)$, $\CSP(\Gamma_2)$.


\section{Strong-Backdoors to Scattered Classes}
\label{sect:csp}



%

This section is dedicated to proving our main technical lemma, restated below. We would like to point out that the assumption regarding the existence of the tautological binary relation $\DDD^2$ in the languages is made purely for ease of description in the later stages of the algorithm.

\addtocounter{lemma}{-1}
\begin{lemma}
\label{lem:main}
Let $\Gamma_1,\dots \Gamma_d$ be finite languages over a finite domain $\DDD$ which are closed under partial assignments and contain $\DDD^2$. Then \textsc{SBD}$(\CSP(\Gamma_1)\uplus\cdots\uplus \CSP(\Gamma_d))$ can be solved in time $2^{2^{\bigoh(k)}}\vert\III\vert^{\bigoh(1)}$.
\end{lemma}

Before proceeding further, we show how Lemma~\ref{lem:main} is used to prove Theorem~\ref{thm:main}.

\begin{proof}[Proof of Theorem~\ref{thm:main}]
Let $\III$ be an instance of $\CSP(\relD)$. Recalling the definition
of $\Gamma^*$, we use Lemma~\ref{lem:main} to find a strong backdoor
$X$ of size at most $k$ into $\CSP(\Gamma_1^*)\uplus\cdots\uplus
\CSP(\Gamma_d^*)$ in time
$2^{2^{\bigoh(k)}}\vert\III\vert^{\bigoh(1)}$. Since
$\CSP(\Gamma_1^*)\uplus\cdots\uplus \CSP(\Gamma_d^*)\supseteq
\CSP(\Gamma_1)\uplus\cdots\uplus \CSP(\Gamma_d)$, it follows that any
strong backdoor into $\CSP(\Gamma_1)\uplus\cdots\uplus \CSP(\Gamma_d)$
is also a strong backdoor into $\CSP(\Gamma_1^*)\uplus\cdots\uplus
\CSP(\Gamma_d^*)$. We branch over all the at most $|\DDD|^k$ assignments $\alpha:X\rightarrow \DDD$, and for each such $\alpha$ we can solve the instance $\III|_\alpha$ in polynomial time since $\CSP(\Gamma_1^*)\uplus\cdots\uplus \CSP(\Gamma_d^*)$ is tractable.

For the second case, let $\III$ be an instance of
$\sharpCSP(\relD)$. As above, we also use Lemma~\ref{lem:main} to
compute a strong backdoor $X$ into $\CSP(\Gamma_1^*)\uplus\cdots\uplus
\CSP(\Gamma_d^*)$ of size at most $k$. We then branch over all at most $|\DDD|^k$ assignments $\alpha:X\rightarrow \DDD$, and for each such $\alpha$ we can solve the {\sharpCSP} instance $\III|_\alpha$ in polynomial time since $\CSP(\Gamma_1^*)\uplus\cdots\uplus \CSP(\Gamma_d^*)$ is \sharpp-tractable; let $\text{cost}(\alpha)$ denote the number of satisfying assignments of $\III|_\alpha$ for each $\alpha$. We then output $\sum_{\alpha:X\rightarrow \DDD}\text{cost}(\alpha)$.
\end{proof}

We begin our path towards a proof of Lemma~\ref{lem:main} by stating the following assumption on the input instance, which can be guaranteed by simple preprocessing. Let $\rho$ be the maximum arity of any relation in $\Gamma_1,\dots, \Gamma_d$. 

\begin{observation}
\label{obs:arity}
Any instance $(\III',k)$ of \textsc{SBD}$(\CSP(\Gamma_1)\uplus\cdots\uplus \CSP(\Gamma_d))$ either contains only constraints of arity at most $\rho+k$, or can be correctly rejected.
\end{observation}

\lv{
\begin{proof}
Assume that $\III'$ contains a constraint $C=(S,R)$ of arity $\rho'>\rho+k$. Then for every set $X$ of at most $k$ variables, there exists an assignment $\alpha:X\rightarrow \DDD$ such that $C|_\alpha$ has arity $\rho'>\rho$, and hence $C|_\alpha\not \in \CSP(\Gamma_1)\uplus\cdots\uplus \CSP(\Gamma_d)$. Hence any such $(\III',k)$ is clearly a {\No}-instance of \textsc{SBD}($\CSP(\Gamma_1)\uplus\cdots\uplus \CSP(\Gamma_d)$).
\end{proof}}

\lv{\paragraph{Organization of the rest of the section.}
The rest of this section is structured into three subsections. In Subsection~\ref{sub:ic}, we use  iterative compression to transform the \textsc{SBD} problem targeted by Lemma~\ref{lem:main} into its compressed version \ebdcomp. Subsection~\ref{sub:nonsep} develops an algorithm which correctly solves any instance of {\ebdcomp} which has a certain inseparability property. Finally, in Subsection~\ref{sub:general} we give a general algorithm for {\ebdcomp} which uses the algorithm developed in Subsection~\ref{sub:nonsep} as a subroutine.}

\subsection{Iterative compression}
\label{sub:ic}
We first describe a way to reduce the input instance of {\strongbds} to multiple (but a bounded number of) \emph{structured} instances, such that solving these instances will lead to a solution for the input instance. To do this, we use the technique of iterative compression \cite{ReedSV04}. 
\lv{Given an instance $(\III,k)$ of {\strongbds} where $\III=\{C_1,\dots,C_{m}\}$, for $i\in [m]$ we define $\bC_i=\{C_1,\dots,C_i\}$. We iterate through the instances $(\bC_i,k)$ starting from $i=1$, and for each $i$-th instance we use a \emph{known} solution $X_i$ of size at most $k+\rho$ to try to find a solution $\hat X_i$ of size at most $k$. This problem, usually referred to as the \emph{compression} problem, is the following.

  \begin{center}
  \begin{boxedminipage}[t]{0.99\textwidth}
  \begin{quote}
  \textsc{{\strongbds} Compression}\\ \nopagebreak
  \nopagebreak
  \emph{Setting}: Languages
  $\Gamma_1,\dots,\Gamma_d$ of maximum arity $\rho$ over a domain $\DDD$.\\ \nopagebreak
  \emph{Instance}: A CSP instance $\III$, a non-negative integer $k$ and a strong backdoor set $X\subseteq \var(\III)$ into $\CSP(\Gamma_1)\uplus\cdots\uplus\CSP(\Gamma_d)$ of size at most $k+\rho$.\\ \nopagebreak
  \emph{Task}: Find a strong backdoor in $\III$ into $\CSP(\Gamma_1)\uplus\cdots\uplus\CSP(\Gamma_d)$ of size at most $k$, or correctly
   determine that no such set exists.\\ \nopagebreak
  \emph{Parameter}: $k$.
\end{quote}
\end{boxedminipage}
\end{center}

\medskip

When $\Gamma_1,\dots,\Gamma_d$ are clear from the context, we abbreviate 
{\strongbds} as
\renewcommand{\strongbds}{\textsc{SBD}}
{\strongbds} and \textsc{\strongbds$(\CSP(\Gamma_1)\uplus\cdots\uplus\CSP(\Gamma_d))$ Compression} as \bdcomp. We reduce the {\strongbds} problem to $m$ instances of the {\bdcomp} problem as follows. 
Let $\III'$ be an instance of \strongbds.
The set $\var(C_1)$ is clearly a strong backdoor of size at most $\rho$ for the instance $\III_{1}=(\bC_1,k, \emptyset)$ of \bdcomp. We construct and solve a sequence of {\bdcomp} instances $\III_{2},\dots \III_{m}$ by letting $\III_i=(\bC_i,k,X_{i-1}\cup \var(C_i))$, where $ X_{i-1}$ is the solution to $\III_{i-1}$. If some such $\III_i$ is found to have no solution, then we can correctly reject for $\III'$, since $\bC_i\subseteq \III'$. On the other hand, if a solution $X_m$ is obtained for $\III_m$, then $X_m$ is also a solution for $\III'$. Since there are $m$ such iterations, the total time taken is bounded by $m$ times the time required to solve the {\bdcomp} problem. 


\paragraph{Moving from the compression problem to the extension version.}
We now show how to convert an instance of the {\bdcomp} problem into a bounded number of instances of the same problem where we may additionally assume the solution we are looking for extends part of the given strong backdoor. Formally, an instance of the {\sc Extended} {\bdcomp} problem is a tuple $(\III,k,S,W)$ where $\III$ is a CSP instance, $k$ is a non-negative integer and $W\cup S$ is a strong backdoor set of size at most $k+\rho$ into {\splitclass}. The objective here is to compute a strong backdoor into {\splitclass} of size at most $k$ which contains $S$ and is disjoint from $W$. 
\begin{center}
  \begin{boxedminipage}[t]{0.99\textwidth}
  
\begin{quote}
  \textsc{Extended {\strongbds} Compression} ({{\sc Ext-}\bdcomp})\\ \nopagebreak
  \nopagebreak
  \emph{Setting}: Languages
  $\Gamma_1,\dots,\Gamma_d$ of maximum arity $\rho$ over a domain $\DDD$.\\ \nopagebreak
  \emph{Instance}: A CSP instance $\III$, a non-negative integer $k$ and disjoint variable sets $S$ and $W$ such that $W\cup S$ is a strong backdoor set into {\splitclass} of size at most $k+\rho$.\\ \nopagebreak
  \emph{Task}: Find a strong backdoor in $\III$ into 
   {\splitclass} of size at most $k$ that extends $S$ and is disjoint from $W$, or
   determine that no such strong backdoor set exists.\\ \nopagebreak
  \emph{Parameter}: $k$.
\end{quote}
\end{boxedminipage}
\end{center}

We now reduce {\bdcomp} to $\vert X\vert  \choose {\leq k}$-many instances of {\ebdcomp}, as follows. Let $\III'=(\III,k,X)$ be an instance of {\bdcomp}. We construct $\vert X\vert  \choose {\leq k}$-many instances of {\ebdcomp} as follows. For every $S\in {X\choose {\leq k}}$, we construct the instance $\III'_S=(\III,k,S,X\setminus S)$. Clearly, the original instance $\III'$ is a {\Yes} instance of {\bdcomp} if and only if for some $S\in {X \choose {\leq k}}$, the instance $\III'_S$ is a {\Yes} instance of {\ebdcomp}. Therefore, the time to solve the instance $\III'$ is bounded by ${\vert X\vert  \choose {\leq k}}\leq 2^{k+\rho}$ times the time required to solve an instance of {\ebdcomp}. In the rest of the paper, we give an FPT algorithm to solve {\ebdcomp}, which following our discussion above implies Lemma~\ref{lem:main}.}
\sv{This allows us to reduce {\strongbds} to the {\ebdcomp} problem defined below in a way that an FPT algorithm for the latter implies one for the former.
\begin{center}
   \vspace{-0.5cm}
  \begin{boxedminipage}[t]{0.99\textwidth}
\begin{quote}
  \textsc{Extended {\strongbds} Compression} ({{\sc Ext-}\bdcomp})\\ \nopagebreak
  \nopagebreak
  \emph{Setting}: Languages
  $\Gamma_1,\dots,\Gamma_d$ of maximum arity $\rho$ over a domain $\DDD$.\\ \nopagebreak
  \emph{Instance}: A CSP instance $\III$, a non-negative integer $k$ and disjoint variable sets $S$ and $W$ such that $W\cup S$ is a strong backdoor set into {\splitclass} of size at most $k+\rho$.\\ \nopagebreak
  \emph{Task}: Find a strong backdoor in $\III$ into 
   {\splitclass} of size at most $k$ that extends $S$ and is disjoint from $W$, or
   determine that no such strong backdoor set exists.\\ \nopagebreak
  \emph{Parameter}: $k$.
\end{quote}
\end{boxedminipage}
\end{center}
In the rest of the paper, our focus will be on proving the following lemma.
}

\begin{tlemma}\label{lem:ebdcomp algo}
{\ebdcomp} can be solved in time $2^{2^{\bigoh(k)}}\vert\III\vert^{\bigoh(1)}$.
\end{tlemma}

We first focus on solving a special case of {\ebdcomp}, and then show how this helps to solve the problem in its full generality.

\subsection{Solving non-separating instances}
\label{sub:nonsep}
In this subsection, we restrict our attention to input instances with a certain promise on the structure of a solution. We refer to these special instances as \emph{non-separating instances}. These instances are formally defined as follows.

\begin{figure}[t]
\begin{center}
\includegraphics[height=240 pt, width=350 pt]{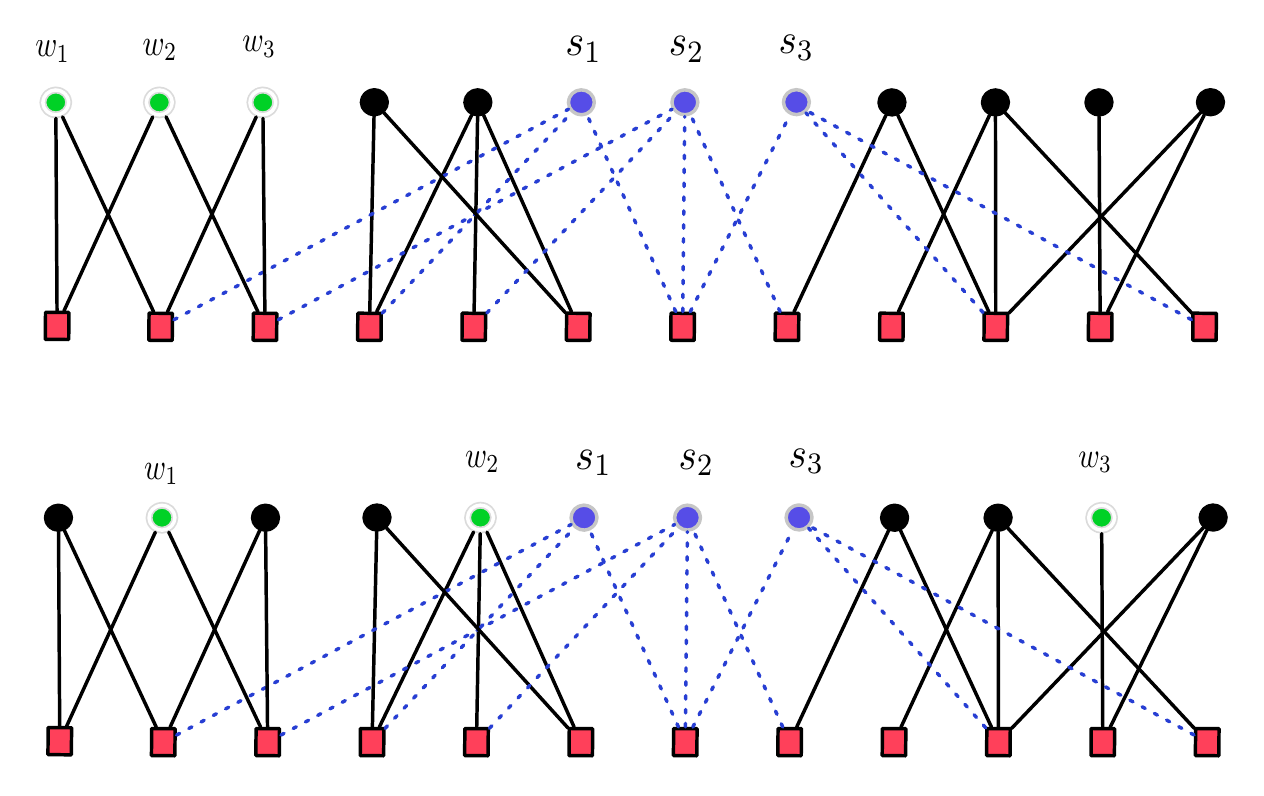}
\caption{An illustration of separating and non-separating solutions. In both cases, $S=\{s_1,s_2,s_3\}$ is the hypothetical solution under consideration while $\{w_1,w_2,w_3\}$ is the old solution. In the first figure, $S$ is a non-separating solution while in the second, it is a separating solution. }
\label{fig:non separating solution}
\end{center}
\end{figure}

\begin{definition}
Let $(\III,k,S,W)$ be an instance of {\ebdcomp} and let $Z\supseteq S$ be a solution for this instance. We call $Z$ a {\bf separating solution} (see Figure \ref{fig:non separating solution}) for this instance if $W$ is not contained in a single connected component of $\cB_Z$ and a {\bf non-separating solution} otherwise. 
An instance is called a {\bf separating instance} if it only has separating solutions and it is called a {\bf non-separating instance} otherwise.
%
%
\end{definition}

Having formally defined non-separating instances, we now give an overview of the algorithm we design to solve such instances. We begin by developing the notion of a \emph{forbidden set of constraints}. The main motivation behind the introduction of this object is that it provides us with a succinct certificate that a particular set is \emph{not} a strong backdoor of the required kind, immediately giving us a small structure which we must exclude. However, the exclusion in this context can occur not just by instantiating a variable in the scope of one of these constraints in the solution but also due to the backdoor disconnecting these constraints. This is significantly different from standard graph problems where once we have a small violating structure, a straightforward branching algorithm can be used to eliminate this structure. However, in our case, even if we have a small violating structure, it is not at all clear how such a structure can be utilized. For this, we first set up appropriate separator machinery for CSP instances. We then argue that for any forbidden set of constraints, if a variable in the scope of these constraints is not in the solution, then one of these constraints must in fact be separated from the  old strong backdoor set by the hypothetical solution. Following this, we argue that the notion of \emph{important separators} introduced by Marx \cite{Marx06} can be used to essentially narrow down the search space of separators where we must search for a solution variable. Finally, we can use a branching algorithm in this significantly pruned search space of separators in order to compute a solution (if there exists one). We reiterate that the notion of forbidden sets is critical in obtaining an FPT algorithm as opposed to an FPT-approximation algorithm.
Now that we have given a slightly more detailed overview of this subsection, we proceed to describe 
our algorithm for solving non-separating instances. We begin with the definition of forbidden constraints and then set up the separator machinery required in this as well as the next subsection.

\paragraph{Forbidden Constraints.} 
  Let $(\III,k,S,W)$ be an instance of {\ebdcomp}.

\begin{definition}
Let $S\subseteq \var(\III)$, let ${\bC}=\{C_1,\dots, C_\ell\}$ be a set of at most $d$ constraints and $J$ be a subset of $[d]$. We say that $\bC$ is $J$-{\bf forbidden with respect to $S$} if there is an assignment $\tau:S\rightarrow \DDD$ such that for every $i\in J$ there is a $t\in [\ell]$ such that $C_t[\tau]\notin \Gamma_i$. If $J=[d]$, then we simply say that $\C$ is forbidden with respect to $S$. Furthermore, we call $\tau$ an assignment \textbf{certifying} that $\C$ is $J$-forbidden (forbidden if $J=[d]$) with respect to $S$. 
\end{definition}

\lv{The following observation is a consequence of the languages being closed under partial assignments.

\begin{observation}\label{obs:local certificate}
If $\bC$ is a set of constraints forbidden with respect to a variable set $S$, then $\bC$ is also forbidden with respect to the set $S\cap \var(\bC)$. Conversely, if $\bC$ is forbidden with respect to $S$ and $S'$ is a set of variables disjoint from $\var(\bC)$, then $\bC$ is also forbidden with respect to $S\cup S'$ and with respect to $S'$.
\end{observation}

}
The intuition behind the definition of forbidden sets is that it allows us to 
have succinct certificates for non-solutions. This intuition is formalized in the following lemma.

\begin{lemma}\label{lem:equivalence} Given a CSP instance $\III$, a set $X\subseteq \var(\III)$ is a strong backdoor set into {\splitclass} if and only if there is no connected component of $\cB_X$ containing a set of constraints forbidden with respect to~$X$.
\end{lemma}

\lv{
\begin{tlemma}\label{lem:forbidden check} Given a CSP instance $\III$ and a set $S$ of variables, we can check in time $\bigoh(\vert \cD\vert^{\vert S\vert}\cdot \vert I\vert^{\bigoh(1)})$ if there is a set of constraints forbidden with respect to $S$. 
 \end{tlemma}

\begin{proof}
Clearly, it is sufficient to run over all at most $d$-sized sets of constraints and all assignments to the variables in $S$ and examine the reduced constraints if they belong to each of the languages $\Gamma_1,\dots, \Gamma_d$. Since these languages are finite, the final check can be done in time $\bigoh(1)$.
 This completes the proof of the lemma.
\end{proof}}

\lv{
\begin{tlemma}\label{lem:branching}Let $\III$ be a CSP instance and let $\C$ be a set of constraints contained in a component of $\III$ and forbidden (with respect to some set). Let $Z$ be a strong backdoor set to {\splitclass} for this instance. Then, either $Z$ disconnects $\C$ or $Z\cap \var(\C)\neq \emptyset$.
\end{tlemma}

\begin{proof}
Suppose that $\C$ occurs in a single component of $\cB_Z$ and $Z$ is disjoint from $\var(\C)$. 
By Observation \ref{obs:local certificate},
 we have that $\C$ is also forbidden with respect to 
any set of variables disjoint from $\var(\C)$, and in particular with respect to $Z$. By Lemma \ref{lem:equivalence}, this
contradicts our assumption that $Z$ is a solution. This completes the proof of the lemma.
\end{proof}}

\paragraph{Separators in CSP instances.} For a variable set $X$, we denote by $\C(X)$ the set of all constraints whose scope has a non-empty intersection with $X$ (equivalently, $\C(X)$ contains the neighbors of $X$ in $\cB$).

\begin{definition}Let $\III$ be an instance of CSP. Let $X,S\subseteq \var(\III)$ be disjoint sets of variables, where $X\cup \C(X)$ is connected in $\cB=\cB(\III)$. We denote by $R_{\cB}(X,S)$ the set of variables and constraints which lie in the component containing $X$ in $\cB- S$ and by $R_{\cB}[X,S]$ the set $R_{\cB}(X,S)\cup S$. Similarly, we denote by $NR_{\cB}(X,S)$ the set $(\var(\III)\cup \III) \setminus R_{\cB}[X,S]$ and by $NR_{\cB}[X,S]$ the set $NR_{\cB}(X,S)\cup S$. We drop the subscript $\cB$ if it is clear from the context.\end{definition}

\begin{definition}\label{vertex separator} Let $\III$ be an instance of CSP and let $\cB=\cB(\III)$. Let $X$ and $Y$ be disjoint variable sets. 
\begin{itemize}
\setlength\itemsep{0em}
\item A variable set $S$ disjoint from $X$ and $Y$ is said to \textbf{disconnect} $X$ and $Y$ (in $\cB$) if $R_{\cB}(X,S)\cap Y=\emptyset$. 
\item If $X\cup \C(X)$ is connected, and $S$ disconnects $X$ and $Y$, then we say that $S$
is an $X$-$Y$ \textbf{separator} (in $\cB$). 
\item An $X$-$Y$ separator is said to be \textbf{minimal} if none of its proper subsets is an $X$-$Y$ separator. 
\item An $X$-$Y$ separator $S_1$ is said to \textbf{cover} an $X$-$Y$ separator $S$ with respect to $X$ if $R(X,S_1)\supset R(X,S)$. If the set $X$ is clear from the context, we just say that $S_1$ covers $S$. 

\item Two $X$-$Y$ separators $S$ and $S_1$ are said to be \textbf{incomparable} if neither covers the other.

\item In a set $\cH$ of $X$-$Y$ separators, a separator $S$ is said to be \textbf{component-maximal} if there is no separator $S'$ in $\cH$ which covers $S$.
Component-minimality is defined analogously.

\item An $X$-$Y$ separator $S_1$ is said to \textbf{dominate} an $X$-$Y$ separator $S$ with respect to $X$ if $\vert S_1\vert \leq \vert S\vert$ and $S_1$ covers $S$ with respect to $X$. If the set $X$ is clear from the context, we just say that $S_1$ dominates $S$. 

 \item We say that $S$ is an \textbf{important} $X$-$Y$ separator if it is minimal and there is no $X$-$Y$ separator dominating $S$ with respect to $X$. 
 \end{itemize}
\end{definition}

Note that we require separators to only occur in the variable set of $\cB$, as is reflected in the definitions above; this differs from the standard graph setting in \cite{Marx06,ChenLL09}. However, it is known that their results also carry over to this more general setting. Specifically, for any graph $G=(V,E)$ and any $A,X,Y\subseteq V$, it is possible to construct a supergraph $G'\supseteq G$ such that the set of all important $X$-$Y$ separators in $G$ of size $k$ which are disjoint from $A$ is exactly the set of all important $X$-$Y$ separators of size $k$ in $G'$ (it suffices to make $k+1$ copies of each vertex in $A$). This allows the direct translation of the following results into our setting. 
\sv{\enlargethispage*{5mm}}
Lemma~\ref{lem:number of imp seps}, which is implicit in~\cite{ChenLL09}, plays a crucial role in our algorithm to compute non-separating solutions.

\begin{lemma}\label{lem:number of imp seps}\label{bound undirected}{\rm \scite{ChenLL09}} For every $k\geq 0$ there are at most $4^k$ important $X$-$Y$ separators of size at most $k$. Furthermore, there is an algorithm that runs in time $\bigoh(4^kk\vert \III\vert)$ which enumerates all such important $X$-$Y$ separators, and there is an algorithm that runs in time $|\III|^{\bigoh(1)}$ which outputs one arbitrary component-maximal $X$-$Y$ separator.
\end{lemma}

\lv{
\begin{observation}\label{obs:dominating separator}Let $S_1$ and $S_2$ be two minimal $X$-$Y$ separators where $S_2$ dominates $S_1$. Then, $S_2$ disconnects $(S_1\setminus S_2)$ and $Y$.
\end{observation}
}

We now proceed to the description of our algorithm to solve non-separating instances.
\paragraph{Computing non-separating solutions.} We begin with the following preprocessing rule which can be applied irrespectively of the existence of a non-separating solution.

\begin{reduction}\label{red:irrelevant}
Let $(\III,k,S,W)$ be an instance of {\ebdcomp}. If a connected component of $\cB_S$ does not contain a set of constraints forbidden with respect to $S$ then remove the constraints and variables contained in this connected component.
\end{reduction}

\lv{
\begin{tlemma}\label{lem:connectedtow}
The above preprocessing rule is correct and can be applied in time $\bigoh(\vert \cD\vert^{\vert S\vert} \vert \III\vert^{\bigoh(1)} )$. Furthermore, if the rule does not apply, every connected component of $\cB_S$ intersects $W$.
\end{tlemma}

\begin{proof}
It follows from Lemma \ref{lem:forbidden check} that the rule can be applied in time $\bigoh(\vert \cD\vert^{\vert S\vert} \vert\III\vert^{\bigoh(1)})$. We now argue the correctness of the rule. Let $\cX$ be the component of $\cB_S$ removed by an application of the reduction rule and let $(\III',k,S,W')$ be the resulting reduced instance of {\ebdcomp}. Since $\III'$ is an induced sub-instance of $\III$, any solution $Z$ for $(\III,k,S,W)$ also represents a solution $Z\setminus \cX$ for $(\III',k,S,W')$. For the converse direction, consider a solution $Z'\supseteq S$ of $(\III',k,S,W')$. By Lemma~\ref{lem:equivalence}, this implies that there is no component in $\cB(\III')_{Z'}$ which contains a set of constraints forbidden with respect to $Z'$. Since $\cX$ also contains no set of constraints forbidden with respect to $Z'$, it follows that $Z'$ is also a solution for $(\III,k,S,W)$.
This completes the proof of correctness of the preprocessing rule.
 We now prove the final statement of the lemma.

Suppose that the rule is no longer applicable and there is a component $\cX$ of $\cB_S$ disjoint from $W$. Since the rule is not applicable, there is a set $\C$ of constraints in $\cX$ forbidden with respect to $S$. Since $\cX$ is disjoint from $W$, we have that $\C$ is contained in a single component of $\cB_{W\cup S}$. Furthermore, since $W\cap \var(\C)=\emptyset$, it must be the case that $\C$ is also forbidden with respect to $W\cup S$ (by Observation \ref{obs:local certificate}), a contradiction to the assumption that $W\cup S$ is a strong backdoor set for the given CSP instance. This completes the proof of the lemma.
\end{proof}}

For the following, we slightly expand the defined notions of \emph{separators} and in particular \emph{important separators} by specifying a subgraph of $\cB$ which these will operate in.

\begin{tlemma}\label{lem:pushing}
Let $(\III,k,S,W)$ be an instance of {\ebdcomp} and let $Z$ be a non-separating solution for this instance. Furthermore, let $v$ be a variable such that $Z$ disconnects $v$ and $W$. Then there is a solution $Z'$ which  contains an important $v$-$W$ separator of size at most $k$ in $\cB_S$.
\end{tlemma}

\lv{
\begin{proof}
By Lemma \ref{lem:connectedtow}, we have that there is a component of $\cB_S$ containing $v$ and intersecting $W$. Therefore, it must be the case that $Z\setminus S$ contains a non-empty set, say $A$ which is a minimal $v$-$W$ separator of size at most $k$ in $\cB_S$. If $A$ is an important $v$-$W$ separator in $\cB_S$ then we are done. Suppose that this is not the case. Then there is a $v$-$W$ separator in $\cB_S$, say $B$, which dominates $A$. We claim that the set $Z'=(Z\setminus A)\cup B$ is also a solution for the given instance. 

Clearly, $Z'$ is no larger than $Z$, contains $S$ and is disjoint from $W$. It remains to show that $Z'$ is also a solution. By Lemma \ref{lem:equivalence}, it suffices to show that there is no connected component of $\cB_{Z'}$ which contains a set of constraints forbidden with respect to $Z'$. 

Suppose that there is a component of $\cB_{Z'}$ containing a set $\C$ of constraints forbidden with respect to $Z'$. We first consider the case when this component, say $\cal X$, is disjoint from $W$. Since $S\subseteq Z'$, it must be the case that there is a component $\cY$ of $\cB_{W\cup S}$ such that $\cX\subseteq \cY$. Furthermore, since $Z'$ is disjoint from $W$, we have that $\var(\C)$ is disjoint from $W$. Therefore, we conclude that $\C$ is forbidden with respect to $W\cup S$ (by Observation \ref{obs:local certificate}). Since we have already argued that $\C$ is contained in a single component of $\cB_{W\cup S}$, we infer that $W\cup S$ is not a strong backdoor for the given instance, a contradiction. Therefore, we conclude that $\cX$ must intersect $W$. Notice that this rules out the possibility of $\C$ being contained in the set $R_{\cB_S}(v,B)$, since this set is by definition disconnected from $W$ by $B\cup S$, which is a subset of $Z'$.

By the definition of $Z'$, it follows that any component of $B_{Z'}$ intersecting $Z\setminus Z'$ is contained in the set $R_{\cB_S}(v,B)$. This implies that the component $\cX$ which contains $\C$ is in fact disjoint from $Z\setminus Z'$. Hence, it must be the case that there is a connected component, say $\cal H$ in $\cB_Z$ such that $\cX\subseteq \cH$. Now, let $\tau:Z'\rightarrow \cD$ be an assignment to the variables in $Z'$ which certifies that $\C$ is forbidden w.r.t $Z'$. Let $\hat Z= Z\cap Z'$ and $\tau'=\tau|_{\hat Z}$. Observe that since the languages $\Gamma_1,\dots, \Gamma_d$ are closed under partial assignments, for every constraint $c\in \C$, it must be the case that $c[\tau]\subseteq c[\tau']$. Therefore, $\C$ is also forbidden with respect to $\hat Z$ and  $\tau'$ is an assignment that certifies this.
Furthermore, since $Z\setminus \hat Z$ is disjoint from $\var(\C)$, by Observation \ref{obs:local certificate}, $\C$ is forbidden with respect to $\hat Z\cup (Z\setminus \hat Z)=Z$.
Furthermore,  since $\C$ is contained in  $\cX$ and hence in $\cH$, we conclude that $\C$ is a set of constraints contained in a single component of $\cB_Z$ and forbidden with respect to $Z$, which by Lemma \ref{lem:equivalence}, results in a contradiction to our assumption that $Z$ is a solution for the given instance.
%
 This completes the proof of the lemma.
\end{proof}}

We use the above lemma along with Lemma \ref{lem:number of imp seps} to obtain our algorithm for non-separating instances.

\begin{lemma}\label{lem:type1 solution}
Let $(\III,k,W,S)$ be a non-separating instance of {\ebdcomp}. Then it can be solved in time $2^{\bigoh(k^2)}\vert \III\vert^{\bigoh(1)}$.
\end{lemma}

\sv{\begin{proof}[Sketch of Proof]
Let $\C$ be a set of constraints forbidden w.r.t. $S$ in some component of $\cB_S$. The algorithm recursively branches by trying to either add one variable from $\var(\C)$ into $S$, or for each choice of $v\in \var(\C)$ branches over all variables in the union of all important $v$-$W$ separators of size at most $k$ in $\cB_S$.

The correctness follows from the following claims. Let $Z$ be a non-separating solution. First, either $\C$ is disconnected by $Z$ but not by $S$, or $Z\setminus S$ intersects $\var(\C)$. In the latter case, we will find some variable in $Z\setminus S$ and add it to $S$. In the former case, the second claim is that $\C$ must be disconnected from $W$ by $Z$. Then by Lemma~\ref{lem:pushing},  branching over all important $v$-$W$ separators of size at most $k$ is sufficient to hit some solution $Z'$.
\end{proof}}

\lv{
\begin{proof}  
We first apply Preprocessing Rule \ref{red:irrelevant} exhaustively.
If there is no set of constraints forbidden with respect to $S$ left in a single component of $\cB_S$ after the exhaustive preprocessing, then we are done and $S$ itself is the required solution. Similarly, if $k=\vert S\vert$, then we check if $S$ itself is the solution and if not we return {\No}. Otherwise, let $\C$ be a set of constraints forbidden with respect to $S$ and contained in a component of $\cB_S$. 
We now branch in $\vert\var(\C)\setminus S\vert$-ways by going over all variables $x\in \var(\C)\setminus S$ and in each branch recurse after adding the corresponding variable to $S$. Then, for every $v\in \var(\C)$, we similarly branch on all variables contained in the union of all important $v$-$W$ separators of size at most $k$ in $\cB_S$. This completes the description of the algorithm.
 
We now prove the correctness of the algorithm as follows.
Let $Z\supseteq S$ be a non-separating solution for the given instance. By Lemma \ref{lem:branching}, either $\C$ is disconnected by $Z$ or $Z$ intersects $\var(\C)$. Suppose that $\C$ is not disconnected by $Z$. Observe that it cannot be the case that $Z\cap \var(\C)=S$ since this would then imply that $\C$ is also forbidden with respect to $Z$, a contradiction. Therefore, if $\C$ is not disconnected by $Z$, then there is a variable, say $x$, in $Z\cap \var(\C)$ which is not in $S$. Since we have branched on all the variables in $\var(\C)\setminus S$, we will have a `correct' branch where we have added $x$ to the solution.

 We now describe how the algorithm accounts for the case when $Z$ does not intersect $\var(\C)$.
In this case, $\C$ is disconnected by $Z$. By Lemma \ref{lem:connectedtow}, the constraints in $\C$ are connected to $W$ in $\cB_S$. Since $Z$ is a non-separating solution, there is a constraint, say $C\in\C$, which is disconnected from $W$ by $Z$. Since we have already excluded $\var(C)\subseteq Z$, there must be a variable $v$ in $\var(C)$ which is disconnected from $W$ by $Z$. Then by Lemma~\ref{lem:pushing}, we know that there is also a solution for the given instance which contains an important $v$-$W$ separator of size at most $k$ in $\cB_S$. However, we have branched on the variables in the union of $v$-$W$ separators for every $v\in \var (C)$ for every choice of $C$. This completes the proof of correctness of the algorithm.

We now bound the running time as follows.
For the first round of branchings, since $\vert \var(\C)\vert$ is bounded by $(\rho+k)d$, we have $\bigoh(k)$ branches. For the second set of branchings, we have $\bigoh(k\cdot 4^k)$ branches (due to Lemma \ref{lem:number of imp seps}). 
Since $S$ is strictly increased whenever we branch, the depth of the search tree is bounded by $k$. We spend time  $2^{\bigoh(k)}\vert \III\vert^{\bigoh(1)}$ (due to Lemma \ref{lem:connectedtow} and Lemma \ref{lem:number of imp seps}) at each node of the search tree, and so the bound on the running time follows. This completes the proof of the lemma.
\end{proof}}

\subsection{Solving general instances}
\label{sub:general}
In this subsection, we describe our algorithm to solve general instances of {\ebdcomp} by using the algorithm to check for non-separating solutions as a subroutine. Essentially, this phase of our algorithm is a more powerful version of the algorithm described in the previous subsection. The main idea behind this part of the algorithm is the following. Since $W$ (the old solution) has size bounded by $k+\rho$, we can efficiently `guess' a partition of $W$ as $(W_1,W_2)$ where $W_1\subset W$ is exactly the subset of $W$ which occurs in a particular connected component after removing some hypothetical solution $S$. Once we guess $W_1$ and $W_2$, we know that the solution we are looking for separates $W_1$ and $W_2$. However, while it is tempting to narrow our search space down to important $W_1$-$W_2$ separators at this point, it is fairly easy to see that such an approach would be incorrect. However, while we are not able to narrow our search space of $W_1$-$W_2$ separators to only important $W_1$-$W_2$ separators, we show that it is indeed possible to prune the search space down to a set  of separators which is much larger than the set of important separators, but whose size is bounded by a function of $k$ nevertheless. Once we do that, the rest of the algorithm is a branching algorithm searching through this space. The main technical content in this part of our algorithm lies in showing that it is sufficient to restrict our search to an \emph{efficiently computable} bounded set of separators. We next give a brief description of the approach we follow to achieve this objective.
  
  At a high level, we use the approach introduced in \cite{LokshtanovR12}. However, there are significant obstacles that arise due to the fact that we are dealing with  scattered classes of CSPs.
  The crux of the idea is the following. We define a laminar family of $W_1$-$W_2$ separators which have a certain monotonicity property. Informally speaking, we partition the separators into `good' and 'bad' separators so that (under some ordering) all the good separators occur continuously followed by  all the bad separators.  Following this, we pick the middle separators in this family---the `last' good separator and the `first' bad separator---and show that deleting either of these separators must necessarily disconnect the hypothetical solution we are attempting to find. Roughly speaking, once we have computed the laminar family of separators,  we delete the middle separators and perform the same procedure recursively on the connected component intersecting  $W_1$. Since the solution has size bounded by $k$, it cannot be broken up more than $k$ times and hence the number of levels in the recursion is also bounded by $k$. We then show that essentially the union of the middle separators computed at the various levels of recursion  of this algorithm has size $f(k)$ and furthermore it is sufficient to restrict our search for a $W_1$-$W_2$ separator to this set.

We begin by defining a \emph{connecting gadget} which consists of redundant constraints and whose purpose is purely to encode connectivity at crucial points of the algorithm.

\begin{tdefinition}
Let $\III$ be a CSP instance and let $X=\{x_1,\dots,x_\ell\}$ be a set of variables. Let $\III'$ be the instance obtained from $\III$ as follows. Add $\ell -1$ new tautological binary constraints $T_1,\dots, T_{\ell-1}$ and define the scope of each $T_i$ as $\{x_i,x_{i+1}\}$. We refer to $\III'$ as \emph{the instance obtained from $\III$ by adding the connecting gadget on~$X$}.
\end{tdefinition}

\lv{
\begin{tlemma}\label{lem:connecting equiv}
Let $(\III,k,S,W)$ be an instance of {\ebdcomp} and let $Z$ be a separating solution for this instance. Let $\cX$ be a component of $\cB_Z$ and let $W_1=W\cap \cX$. Let $\III'$ be the CSP instance obtained from $\III$ be adding the connecting gadget on $W_1$. Then $Z$ is also a solution for $(\III',k,S,W)$.
\end{tlemma}

\begin{proof} 
	Let $\cB'$ be the incidence graph of the CSP instance $\III'$. If $Z$ is not a solution for the instance $(\III',k,S,W)$, then there must be a component of $\cB(\III')_Z$ containing a set $\C$ of constraints forbidden with respect to $Z$. Observe that since by assumption the languages $\Gamma_1,\dots, \Gamma_d$ all contain the tautological binary relation, $\C$ is disjoint from the constraints which were added to $\III$ to construct $\III'$. Finally, any set of constraints from $\III$ which occur together in the same component of $\cB_Z$ also occur together in the same component of $\cB(\III')_Z$ and vice versa. This implies that $\C$ is a set of constraints forbidden with respect to $Z$ and is contained in a single component of $\cB_Z$, a contradiction. This completes the proof of the lemma.
\end{proof}}

From this point on, we assume that if an instance $(\III,k,S,W)$ of {\ebdcomp} is a \emph{separating} instance, then it will be represented as a tuple $(\III,k,S,W_1,W_2)$ where $W_1\subset W$ and $W_2=W\setminus W_1$ with the connecting gadget added on $W_1$. Note that, since $|W|\leq k+\rho$, we will later on be allowed to branch over all partitions of $W$ into $W_1$ and $W_2$ in time $2^{k+\rho}$.
Our objective now is to check if there is a strong backdoor set for $\III$ extending $S$, disjoint from $W$ and separating $W_1$ from $W_2$. For this, we need to introduce the notions of tight separator sequences and pattern replacement procedures. 

\lv{
\subsubsection{Tight separator sequences }
 }

 \sv{ \paragraph{Tight separator sequences}}
\lv{
Let $\III$ be a set of constraints and let $Y$ be a subgraph of $\cB(\III)$. We use $\III[Y]$ to denote $\III\cap Y$. For expositional clarity, we will usually enforce $\var(\III\cap Y)$ to also lie in $Y$.

\begin{tdefinition}
Let $(\III,k,S,W_1,W_2)$ be an instance of {\ebdcomp}. 
We call a $W_1$-$W_2$ separator $X$ in $\cB_S$ $\ell$-\emph{good} if there is a variable set $K$ of size at most $\ell$ such that $K\cup X\cup S$ is a strong backdoor set in $\III[R_{\cB_S}[W_1,X]\cup S]$ into $\CSP(\Gamma_1) \uplus \cdots \uplus \CSP(\Gamma_d)$, 
and we call it $\ell$-\emph{bad} otherwise. 
\end{tdefinition}

\begin{tlemma}(Monotonocity Lemma)\label{lem:monotone}
Let $(\III,k,S,W_1,W_2)$ be an instance of {\ebdcomp} and 
let $X$ and $Y$ be disjoint $W_1$-$W_2$ separators in $\cB_S$ such that $X$ covers $Y$. If $X$ is $\ell$-\emph{good}, then so is $Y$. Consequently, if $Y$ is $\ell$-\emph{bad}, then so is $X$.
\end{tlemma}

\begin{proof}
Suppose that $X$ is $\ell$-good and let $K$ be a variable set of size at most $\ell$ such that $K\cup X\cup S$ is a strong backdoor set into $\CSP(\Gamma_1) \uplus \cdots \uplus \CSP(\Gamma_d)$ for the sub-instance $\III'=\III[R_{\cB_S}[W_1,X]\cup S]$. Let $K'=K\cap R_{\cB_S}[W_1,Y]$. We claim that $Y$ is $\ell$-good and that $P=K'\cup Y\cup S$ is a strong backdoor set for the sub-instance $\hat \III=\III[R_{\cB_S}[W_1,Y]\cup S]$. 

If this were not the case, then there is a set $\C$ of constraints which are contained in a single component of $\cB(\hat \III)-P$ and are forbidden with respect to the set $P$. Let $\tau:P\rightarrow\cD$ be an assignment that certifies this. Since
 the languages $\Gamma_1,\dots, \Gamma_d$ are closed under partial assignments, we conclude that $\tau|_{K'\cup S}$ is an assignment certifying that $\C$ is also forbidden with respect to $K'\cup S$.

Now, since $\C$ lies in the set $R_{\cB_S}[W_1,Y]$, 
no constraint in $\C$ can have in its scope a variable in $X\cup (K\setminus K')$ ($Y$ disconnects these variables from $\C$). Therefore, by Observation \ref{obs:local certificate}, $\C$ being forbidden with respect to $K'\cup S$ implies that it is also forbidden with respect to $(K'\cup S) \bigcup (X\cup (K\setminus K'))=X\cup K\cup S$.

Finally, since $\C$ lies in a single component of $\cB(\hat \III)-P$ and $(K\setminus K')\cup X$ is disjoint from $R_{\cB_S}[W_1,Y]\cup S$, it must be the case that $\C$ also lies in a single component of $\cB(\III')-(K\cup X\cup S)$. But this results in  a  contradiction to our assumption that 
$K\cup X\cup S$ is a strong backdoor set for the instance $\III'$. This completes the proof of the lemma.
\end{proof}

\begin{tdefinition}\label{def:well domination}
Let $(\III,k,S,W_1,W_2)$ be an instance of {\ebdcomp} and let $X$ and $Y$ be $W_1$-$W_2$ separators in $\cB_S$ such that $Y$ dominates $X$. Let $\ell$ be the smallest $i$ for which $X$ is $i$-good. If $Y$ is $\ell$-good, then we say that $Y$ 
  \textbf{well-dominates} $X$. If $X$ is $\ell$-good and there is no $Y\neq X$ which well-dominates $X$, then we call $X$, \textbf{$\ell$-important}.
\end{tdefinition}

\begin{tlemma}\label{lem:special separators replacement}
Let $(\III,k,S,W_1,W_2)$ be an instance of {\ebdcomp} and let $Z$ be a solution for this instance. Let $P\subseteq Z\setminus S$ be a non-empty minimal $W_1$-$W_2$ separator in $\cB_S$ and let $P'$ be a $W_1$-$W_2$ separator in $\cB_S$ well-dominating $P$. Then there is also a solution for the given instance containing $P'$.
\end{tlemma}

\begin{proof} 
Let $Q=(Z\cap [R[W_1,P])\cup S$. Notice that $Q$ is a strong backdoor set into {\splitclass} in the instance $\III'=\III[R[W_1,P]\cup S]]$. Let $Q'\supseteq P'\cup S$ be a smallest strong backdoor set into {\splitclass} extending $P'\cup S$ in the instance $\hat \III=\III[R[W_1,P']\cup S])$. We claim that $Z'=(Z\setminus  Q)\cup Q'$ is a solution for $(\III,k,S,W_1,W_2)$. By definition of well-domination, $Z'$ is no larger than $Z$. It now remains to prove that $Z'$ is a strong backdoor set for the given instance.
Suppose that $Z'$ is not a strong backdoor set and let $\cX$ be  a connected component of $\cB_{Z'}$ containing a set $\C$ of constraints forbidden with respect to $Z'$. 

We first consider the case when $\cX$ is disjoint from the set $Z\setminus Z'$. Then, there is a component $\cH$ in $\cB_Z$ which contains $\cX$ and hence $\C$. Now, let $\tau:Z'\rightarrow \cD$ be an assignment to the variables in $Z'$ which certifies that $\C$ is forbidden with respect to $Z'$. Let $\hat Z=Z\cap Z'$ and let $\tau'=\tau|_{Z'}$. Observe that since the languages $\Gamma_1,\dots, \Gamma_d$ are closed under partial assignments, for every $C\in \C$, it must be the case that $C[\tau]\subseteq C[\tau']$. Therefore, $\C$ is also forbidden with respect to $\hat Z$ and this is certified by $\tau'$. Now, since $Z\setminus \hat Z$ is disjoint from $\var(\C)$, Observation \ref{obs:local certificate} implies that $\C$ is also forbidden with respect to $\hat Z\cup (Z\setminus \hat Z)=Z$. Finally, since $\cX$ is disjoint from $Z$, we conclude that $\C$ occurs in a single connected component of $\cB-Z$ a contradiction to our assumption that $Z$ is a solution.

We now consider the case when $\cX$ intersects the set $Z\setminus Z'$. By the definition of $Z'$ it must be the case that $\cX$ is contained in the set $R_{\cB_S}(W_1,P')$. Furthermore, since $Z'\setminus Q'$ is disjoint from $R_{\cB_S}(W_1,P')$ and is in fact already disconnected from $R_{\cB_S}(W_1,P')$ by just $P'$, it must be the case that $\cX$ and hence $\C$ is also contained in a single component of $\cB(\hat \III)-Q'$. 
Now, since $\C$ is contained in $R_{\cB_S}(W_1,P')$, we know that $\var(\C)$ is disjoint from $Z'\setminus Q'$. Therefore, we conclude that $\C$ is also forbidden with respect to $Q'$. However, this contradicts our assumption that $Q'$ is a strong backdoor set into {\splitclass} in $\hat \III$.
This completes the proof of the lemma.
\end{proof}

Lemma \ref{lem:pushing} and Lemma \ref{lem:special separators replacement} imply that it is sufficient to compute either a variable separated from $W$ by some solution or a separator well-dominating (the separating part of) some solution. Furthermore, Lemma \ref{lem:special separators replacement} allows us to restrict our attention to $\ell$-important separators for appropriate values of $\ell$. In the rest of this section, we describe a subroutine that runs in FPT time and always achieves one of the above objectives.
We use the notion of tight separator sequences (defined below) to streamline our search for this variable/separator.}


\sv{It is possible to prove that to solve {\ebdcomp}, it is sufficient to compute either a variable separated from $W$ by some solution (allowing the use of Lemma~\ref{lem:pushing}), or a separator which, in a precisely defined sense, supersedes the ``separating part'' of some solution (we call this \emph{well-domination}). In the rest of this section, we describe an FPT subroutine which always achieves at least one of these objectives. To this end, we use the notion of \emph{tight separator sequences}, which are maximal sets of pairwise-disjoint important separators of size at most $k$ where for each pair of separators in the set, one covers the other. A tight separator sequence can in fact be computed in polynomial time.
}

\begin{figure}[t]
\begin{center}
\includegraphics[height=300 pt, width=300 pt]{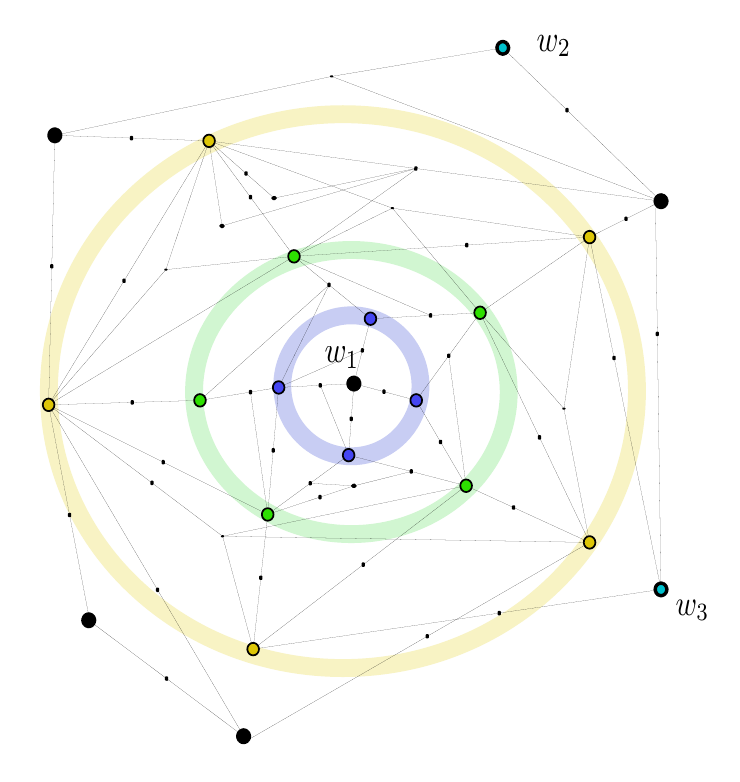}
\caption{An illustration of a tight $w_1$-$\{w_2,w_3\}$ separator sequence. Each shaded layer corresponds to a separator in the sequence}
\label{fig:tight sequence}
\end{center}
\end{figure}


\lv{
\begin{tdefinition}\label{def:smallest separator sequence} 
For every $k\geq 1$, a tight $X$-$Y$ separator sequence of order $k$ is a set $\cH$ of $X$-$Y$ separators with the following properties.

\begin{itemize}
\item Every separator has size at most $k$.
\item The separators are pairwise disjoint.
\item For any pair of separators in the set, one covers the other.
\item The set is maximal with respect to the above properties. 
\end{itemize}
\end{tdefinition}}

See Figure \ref{fig:tight sequence} for an illustration of a tight separator sequence.

\begin{lemma}\label{lem:tight sep computation}
Given a CSP instance $\III$, disjoint variable sets $X$, $Y$, and an integer $k$, a tight $X$-$Y$ separator sequence $\cH$ of order $k$ can be computed in time $\vert \III\vert ^{\bigoh(1)}$.\end{lemma}

\lv{
\begin{proof} If there is no $X$-$Y$ separator of size at most $k$, then we stop the procedure.  Otherwise, 
we compute an arbitrary component-maximal $X$-$Y$ separator of size at most $k$.
This can be done in polynomial time by the algorithm of Lemma \ref{lem:number of imp seps}. We add this separator to the family $\cH$, set $Y:=S$, and iterate this process. We claim that the resulting set is a tight $X$-$Y$ separator sequence of order $k$. It is clear that the first three properties of a tight separator sequence are always satisfied in any iteration. As for the maximality of the set $\cH$ which is finally computed, observe that if there is a separator, say $P$ which can be added to the $\cH$ without violating maximality, then it either contradicts the component-maximality of at least one separator in $\cH$ or it contradicts the termination of the process. The former case occurs if $P$ covers at least one separator in $\cH$ and the latter occurs if $P$ is covered by all separators in $\cH$. 
This completes the proof of the lemma.
\end{proof}}

\lv{
\subsubsection{Boundaried CSPs and Replacements}
}

\sv{\paragraph{Boundaried CSPs and Replacements}}

\begin{definition} A $t$-boundaried CSP instance is a CSP instance  $\III$ with $t$ distinguished labeled variables. The set $\partial(\III)$ of labeled variables is called
the \textbf{boundary} of $\III$ and the variables in $\partial(\III)$ are referred to as the \textbf{terminal variables}.
 Let $\III_1$ and $\III_2$ be two $t$-boundaried CSP instances and let $\mu:\partial(\III_1)\rightarrow \partial(\III_2)$ be a bijection. We denote by
$\III_1\otimes_{\mu} \III_2$ the $t$-boundaried CSP instance obtained by the following \textbf{gluing} operation. We take the disjoint union of the constraints in $\III_1$ and $\III_2$ and then identify each variable $x \in \partial(\III_1)$ with the corresponding variable $\mu(x)\in \partial(\III_2)$.
\end{definition}

\lv{We also define the notion of boundaried CSP instances with  an annotated set of variables. The key difference between the boundary and the annotation is that the annotated set of variables plays no part in gluing operations.
Formally, 

\begin{tdefinition}
A $t$-boundaried CSP instance with an annotated set is a $t$- boundaried CSP instance  $\III$ with a second set of distinguished but unlabeled vertices disjoint from the boundary. The set of annotated vertices is denoted by $\Delta(\III)$.
\end{tdefinition}}

\sv{We now outline how the $t$-boundaried instances defined above are used. Consider an instance of {\ebdcomp} with a solution $Z$ which is disjoint from and incomparable to some $\ell$-good separator $P$. Then some part of $Z\setminus S$, say $K$, lies in $R_{B_S}[W_1,P]\cup S$. We show that, by carefully replacing parts outside of $R_{B_S}[W_1,P]\cup S$ with a small gadget, we can obtain an instance $\III'$ which preserves the part of $K$ inside $R_{B_S}[W_1,P]\cup S$. We also show that some part of $K$ had to lie outside of $R_{B_S}[W_1,P]\cup S$, and hence the solution we seek in $\III'$ is strictly smaller than in $\III$; once we find a solution in~$\III'$, we can use it to find one in $\III$. Furthermore, the number of possible boundaried instances which need to be considered for this replacement is bounded by a function of $k$, which allows exhaustive branching.}

\lv{Before proceeding to the technical Lemma~\ref{lem:replacement2}, we give an informal outline of its claim and intended use. Consider an instance of {\ebdcomp} with a solution $Z$ which is disjoint from and incomparable to some $\ell$-good separator $P$. Then some part of $Z\setminus S$, say $K$, lies in $R_{B_S}[W_1,P]\cup S$. We show that, by carefully replacing parts outside of $R_{B_S}[W_1,P]\cup S$ with a small gadget, we can obtain an instance $\III'$ which preserves the part of $K$ inside $R_{B_S}[W_1,P]\cup S$. We also show that some part of $K$ had to lie outside of $R_{B_S}[W_1,P]\cup S$, and hence the solution we seek in $\III'$ is strictly smaller than in $\III$; once we find a solution in $\III'$, we can use it to find one in $\III$. Furthermore, the number of possible boundaried instances which need to be considered for this replacement is bounded by a function of $k$, which allows exhaustive branching.

\begin{observation}\label{obs:incomparability}
Let $S_1$ and $S_2$ be two disjoint and incomparable $X$-$Y$ separators. Then, $S_2^r=R(X,S_1)\cap S_2, S_1^r = R(X,S_2)\cap S_1$ are both nonempty. Furthermore, $S_2^{nr}=S_2\setminus S_2^r$ and $S_1^{r}=S_1\setminus S_1^r$ are also both nonempty.
\end{observation}

\begin{tlemma}\label{lem:replacement2}
Let $(\III,k,S,W_1,W_2)$ be an instance of {\ebdcomp}, and let $Z$ be a solution for this instance.
 Let $Q$ be a minimal part of $Z\setminus S$ separating $W_1$ from $W_2$ in $\cB_S$, let $K=(Z\cap R_{B_S}[W_1,Q])$, and let $\ell=\vert K\setminus Q\vert$.  Let $P$ be a minimal 
$W_1$-$W_2$ separator in $\cB_S$ which is disjoint from $K$ and incomparable with $Q$, let $Q^r$ be $Q\cap R_{\cB_S}(W_1,P)$ and $Q^{nr}=Q\setminus Q^r$. Similarly, let $P^r$ be $P\cap R_{\cB_S}(W_1,Q)$ and $P^{nr}=P\setminus P^r$, and suppose that $W_1\cup P^r$ has the connecting gadget on it. Let $K^r=(K\cap R_{\cB_S}[W_1,P])$. Let $\III_1=\III[R_{\cB_S}[W_1,P]\cup S]$ be a boundaried CSP instance with $P^r\cup S$ as the boundary.

 Then, there exists a $\vert P^r\cup S\vert$-boundaried CSP instance $\hat \III$ with an annotated set of variables, and a bijection $\mu:\partial(\III)\rightarrow P^r\cup S$ such that the glued CSP instance $\III'=\III_1 \otimes_\mu \hat \III$ has the following properties.

\begin{enumerate}
\item The set $W_1\cup P^{nr}\cup S$ is a strong backdoor set into {\splitclass} in the CSP instance $\III'$.

\item The set $Q^r$ is a $\vert K^r \setminus Q^r\vert $-good $W_1$-$P^{nr}$ separator in $\cB_{S\cup \Delta(\hat \III)}(\III')=\cB'_{S\cup \Delta (\hat \III)}$.
 
\item For any $Q'$ which is a $W_1$-$P^{nr}$ separator in $\cB'_{S\cup \Delta (\hat \III)}$ well dominating $Q^r$ in $\III'$, the set $Q'\cup Q^{nr}$ well dominates the $W_1$-$W_2$ separator $Q$ in $\cB_S$. 

\item 
If $v$ is a variable disconnected from $W_1\cup P^{nr}$ by $K^r$ in $\cB'_{S\cup \Delta (\hat \III)}$, then $v$ is in $R(W_1,P)$ and $v$ is disconnected from $W_1\cup W_2$ by $K$ in $\cB_S$.

\item There is a constant $\eta$ and a family $\cH$ of boundaried CSP instances with an annotated variable set  such that $\cH$ contains $\hat \III$,  
has size bounded by $2^{2^{\eta k}}$ and can be computed in time $2^{2^{\eta k}}k^{\bigoh(1)}$.
\end{enumerate}

\end{tlemma}

\begin{proof}
We first describe the instance $\hat \III$ and then prove that it has the properties claimed by the lemma. Let $K^{nr}=(K\setminus K^r)\cup Q^{nr}$. Consider the sub-instance $\III_2=\SB C\in \III \SM \var(C)\subseteq (NR_{\cB_S}(W_1,P)\cap R_{\cB_S}(W_1,Q))\cup P^r\cup K^{nr}\SE$.
In other words, we take the set of constraints containing variables which are either disconnected from $W_1$ by $P$ but not disconnected from $W_1$ by $Q$, or occur in $P^r \cup K^{nr}$. Two vertices $v,w\in \cB_{S\cup K^{nr}}(\III_2)$ are $\III_2$-\emph{connected} if they are connected in $\cB_{S\cup K^{nr}}(\III_2)$ ($\III_2$-connectivity of sets is defined analogously). Notice that $P^r$ is $\III_2$-connected by assumption.

We now perform the following marking scheme on this hypothetical sub-instance $\III_2$. For every assignment $\tau$ of $S\cup K^{nr}$, and for each $J\subseteq [d]$, if there exists a set $\bC$ of constraints such that 
\begin{enumerate}
\item $\bC|_{\tau}$ is $J$-forbidden (w.r.t. $\emptyset$), and
\item $\C\cup P^r$ is $\III_2$-connected,
\end{enumerate}

then we \emph{mark} the constraints in one such set $\bC$. Since $Q$ is $\ell$-good and in particular since $Q\cup K\cup S$ is a strong backdoor set in $\III[R_{\cB_S}[W_1,Q]\cup S]$, every relation occurring in $\bC$ belongs to one of the finite languages $\Gamma_1,\dots,\Gamma_d$ and hence $|\bC|$ does not depend on $k$. Since the number of possible assignments $\tau$ is bounded by $2^{\bigoh(k)}$, we observe that the set $M$ of all marked constraints has cardinality $2^{\bigoh(k)}$. 

To complete our construction of $\hat \III$, we begin by setting $\hat \III=\III_2[M\cup \var(M)\cup P^r\cup S\cup K^{nr}\cup Q^{nr}]$. We then add the connecting gadget on the set $(\var(M)\cup P^r)\setminus (S\cup K^{nr})$. Finally, we define the boundary of $\hat \III$ to be $P^r\cup S$ and define the annotated set $\Delta(\III)$ to be the set $K^{nr}$.
From the bound on $|M|$, it is readily observed that $|\hat \III|\leq 2^{\bigoh(k)}$.
We now prove that $\III'=\III_1\otimes_{identity} \hat \III$ satisfies the properties claimed by the lemma.

\begin{claim}
\label{cl:one}
The set $W_1\cup P^{nr}\cup S$ is a strong backdoor set into {\splitclass} in the CSP instance~$\III'$.\end{claim}

We actually prove a stronger claim, specifically that already $W_1\cup S$ is a strong backdoor set  (into \splitclass) in $\III'$.
Assume the converse; then there exists a set $\bC$ of forbidden constraints in $\III'$ with respect to $W_1\cup S$ which occur in a connected component of $\cB_{W_1\cup S}(\III')$. We first show that $\bC$ must also be connected in $\cB_{W_1\cup S\cup W_2}(\III)$. Consider any path between $v,w\in \var(\bC)$ in $\cB(\III')$, and some constraint $T=(S,\DDD^2)$ on this path which is not present in $\cB(\III)$; observe that $T$ must be a connecting gadget. Let $v',w'$ be the neighbors of $T$. By construction, $v',w'$ are $\III_2$-connected and hence also connected in $\cB_S(\III)$. Since $(W_1\cup S\cup W_2)\cap \cB_S(\III')=\emptyset$, it follows that any such $v',w'$ and consequently also $v,w$ are connected in $\cB_{W_1\cup S\cup W_2}(\III)$.

So, $\bC$ is also connected in $\cB_{W_1\cup S\cup W_2}(\III)$. Since $\var(\bC)\cap W_2=\emptyset$, we conclude that $\bC$ is also forbidden with respect to $W_1\cup S\cup W_2$. This contradicts the fact that $W_1\cup W_2\cup S$ is a solution to our initial instance of \ebdcomp. 

\begin{claim}
\label{cl:two}
The set $Q^r$ is a $\vert K^r \setminus Q^r\vert $-good $W_1$-$P^{nr}$ separator in $\cB_{S\cup \Delta(\hat \III)}(\III')=\cB'_{S\cup \Delta (\hat \III)}$.
\end{claim}

We first prove that $Q^r$ is indeed a separator as claimed. Suppose for a contradiction that there exists a path $\alpha$ in $\cB'_S$ between $v\in W_1$ and $w\in P^{nr}$. By definition, $Q^r$ is a $W_1$-$P^{nr}$ separator in $\cB_S$, and so $\alpha$ must necessarily contain an edge $v'w'$ in $\cB'_{S\cup \Delta (\hat \III)} - Q^r$ which is not in $\cB_S - Q^r$; by construction, this implies that $v'w'$ are $\III_2$-connected. For any such $v'w'$, there hence exists some path $\alpha'$ in $\cB_S(\III_2)\subseteq \cB_S$, and from this it follows that $v,w$ are also connected in $\cB_S - Q_r$. This contradicts the fact that $Q^r$ is a $W_1$-$P^{nr}$ separator in $\cB_S$, and so $Q^r$ must also be a $W_1$-$P^{nr}$ separator in $\cB'_{S\cup \Delta (\hat \III)}$.

Next, we argue that $Q^r$ is $\vert K^r\setminus Q^r\vert$-good in $\cB'_{S\cup K^{nr}}$, and that this is in fact witnessed by $K^r\setminus Q^r$. Indeed, assume for a contradiction that there exists a set $\C$ of constraints in $\III'$ which are connected and forbidden w.r.t. $S\cup K^{nr}\cup \cup Q^r\cup (K^r\setminus Q^r)=S\cup K$. Clearly, none of the constraints in $\C$ contain the relation $\DDD^2$, and hence $\C\subseteq \III$. Furthermore, $\C$ is also connected in $\III$ by the same path argument as in Claim~\ref{cl:one}: any path between $c_1,c_2\in \C$ which contains edges either exists in $\cB_{S\cup K}$, or can be replaced by a new path in $\cB_{S\cup K}$ which uses $\III_2$-connectivity to circumvent connectivity constraints. But since $\var(\C)\cap (S\cup K)$ is the same in $\III$ and $\III'$, we conclude that $\var(\C)$ is forbidden w.r.t. $S\cup K$ in $\III$. This yields a contradiction with $Q$ being $\ell$-good in $\III$ as witnessed by $K\setminus Q$.

\begin{claim}
\label{cl:three}
For any $Q'$ which is a $W_1$-$P^{nr}$ separator in $\cB'_{S\cup \Delta (\hat \III)}$ well dominating $Q^r$ in $\III'$, the set $Q'\cup Q^{nr}$ well dominates the $W_1$-$W_2$ separator $Q$ in $\cB_S$. 
\end{claim}

The set $\hat Q=Q'\cup Q^{nr}$ dominates $Q$ in $\cB_S$ by definition, and therefore it suffices to prove that $\hat Q$ is $\ell$-good in $\cB_S$. Let $Y'$ be the variable set certifying that $Q'$ is $|K^r\setminus Q^r|$-good in $\cB'_{S\cup K^{nr}}$. We claim that $Y=Y'\cup (K^{nr}\setminus Q^{nr})$ certifies that $\hat Q$ is $\ell$-good in $\cB_S$, and argue that $|Y|\leq \ell$. Since $|Y'|\leq |K^r\setminus Q^r|$ by assumption and both $Y'$ and $|K^r\setminus Q^r|$ are disjoint from $K^{nr}\setminus Q^{nr}$, it follows that $|Y|\leq |(K^{r}\setminus Q^r)\cup (K^{nr}\setminus Q^{nr})|=|(K\setminus Q)|=\ell$.

Now, assume for a contradiction that there exists a set $\C$ of constraints in $\III[R_{\cB_S}[W_1,\hat Q]\cup S]$ which are connected and forbidden w.r.t. $\hat Q\cup S\cup Y$. Observe that $\C\not \subseteq \III_1$, since then $\C$ would also be forbidden w.r.t. $S\cup K^{nr}\cup Q'\cup Y'$ and connected in $\cB'_S$, which would contradict our assumption on $Q'$. So it must be the case that $\C_2=\C\cap \III_2$ is non-empty. First, we consider the simpler case where $\C_2\subseteq \hat \III$; by the construction of $\hat \III$ and in particular the addition of the connectivity gadgets, it follows that $\C$ would then also be connected in $\III'$ and hence forbidden in $\III'$ w.r.t. $S\cup K^{nr}\cup Q'\cup Y'$. This excludes the existence of any set $\C$ of forbidden constraints w.r.t. $\hat Q\cup S\cup Y$ such that $\C\subseteq \III'$.

Next, we proceed the general case where $\C_2$ contains a subset of constraints, say $\C^+_2$, which are not contained in $\III'$ and which are pairwise $\III_2$-connected; we refer to any such set $\C^+_2$ as a \emph{leftover} of $\C$.
  Let us now fix $\C$ to be some set of forbidden constraints in $\III$ which has a minimum number of leftovers; in the previous paragraph, we have argued that $\C$ must contain at least one leftover, and we will reach a contradiction by showing that there exists a set $\C'$ of forbidden constraints w.r.t. $\hat Q\cup S\cup Y$ with less leftovers than $\C$.

Let $\tau$ be the assignment of $S\cup \hat Q\cup Y$ certifying that $\C$ is forbidden in $\III$, and let $\tau'$ be the restriction of $\tau'$ to $K^{nr}\cup S$. Let $\C^+_2$ be some leftover of $\C$, and let $J\subseteq [d]$ be such that $\C^+_2|_{\tau}$ is $J$-forbidden w.r.t. $S \cup Y'_2$ in $\III$. Since $\C^+_2$ was not marked during our construction of $\hat \III$, there must exist another ``marked'' set of constraints $\C^M_2$ in $\hat \III$ with the following properties:
\begin{itemize}
\item $\C^M_2|_{\tau'}$ is also $J$-forbidden (w.r.t. $\emptyset$), and
\item $\C^M_2\cup P^r$ is $\III_2$-connected.
\end{itemize}

To finish the argument, let $\C^M=(\C\setminus \C^+_2)\cup \C^M_2$. By construction, $\C^M$ is connected in $\cB_{S\cup \hat Q\cup Y}(\III)$; indeed, $C^M_2$ is pairwise $\III_2$-connected, and is connected to at least one variable $p\in P^r$ which was $\III_2$-connected to $\C^+_2$, which in turn guarantees connectivity to the rest of $\C$. Furthermore, by the construction of $\C^M$ we observe that $\C^M\setminus \C^M_2$ is $([d]\setminus J)$-forbidden w.r.t. $S\cup \hat Q\cup Y$ and $\C^M_2$ is $J$-forbidden w.r.t. $S\cup \hat Q\cup Y$ (because $(S\cup \hat Q\cup Y)\cap \III_2= K^{nr}\cup S$); so, $\C^M$ must be forbidden w.r.t. $S\cup \hat Q\cup Y$ in $\III$. Since $\C^M$ has one less leftover than $\C$, we have reached a contradiction to the existence of $\C$.

\begin{claim}
\label{cl:four}
If $v$ is a variable disconnected from $W_1\cup P^{nr}$ by $K^r$ in $\cB'_{S\cup \Delta (\hat \III)}$, then $v$ is in $R(W_1,P)$ and $v$ is disconnected from $W_1\cup W_2$ by $K$ in $\cB_S$. 
\end{claim}

By construction of $\hat \III$, it is easy to see that $v$ is not separated from $P^r$ by $K^r$ in $\cB'_{S\cup \Delta (\hat \III)}$, and due to the connecting gadget on $W_1\cup P^r$, $v$ is not separated from $W_1$ as well. We now prove the second part of the claim.
Assume for a contradiction that there exists a path $\alpha$ from $v$ to $W_1\cup W_2$ in $\cB_{S\cup K}$. First, consider the case where $\alpha$ does not intersect $\III_2$. Then $\alpha$ must exist in $\III_1$ and in particular also in $\cB'_{S\cup K^{nr}}$, which contradicts our assumption on $v$. 

On the other hand, assume $\alpha$ does intersect $\III_2$. Since $P^r$ is by definition a $W_1$-$(\III_2\setminus P^r)$ separator in $\cB_S$, this means that $\alpha$ must intersect $P^r$; let $a$ be the first vertex in $P^r$ on the path $\alpha$ from $v$. Since $\alpha$ ends in $W_1\cup W_2$, neither of which intersect with $\cB_{S\cup K^{nr}}(\III_2)$, there must be a last vertex $b$ in $\cB(\III_2$) on $\alpha$ (in other words, the path leaves $\cB_{S\cup K^{nr}}(\III_2)$ from $b$ and does not return there). Since $b\not \in Q\subseteq K$ by assumption, it must follow that $b\in P^r$. But then $a,b$ are connected by a connectivity gadget in $\III'$, and hence there is a path of length $2$ between $a$ and $b$ in $\cB'_{S\cup K^{nr}}$ which guarantees the existence of a $v$-$(W_1\cup W_2)$ path $\alpha'$ in $\cB'_{S\cup K^{nr}}$, a contradiction.

\begin{claim}
\label{cl:five}
There is a constant $\eta$ and a family $\cH$ of boundaried CSP instances such that $\cH$ contains $\hat \III$,  
has size bounded by $2^{2^{\eta k}}$ and can be computed in time $2^{2^{\eta k}}k^{\bigoh(1)}$.
\end{claim}

For the proof of the final statement, observe that $\hat \III$ is an instance containing $2^{\bigoh(k)}$ constraints and variables, at most $k$ of which are marked. Since the number of such marked instances is bounded by $2^{2^{\bigoh(k)}}$ and these can be enumerated in time $2^{2^{\bigoh(k)}}$, the proof of the lemma is complete.
\end{proof}

}
\lv{
\subsubsection{The algorithm for general instances}
}
\sv{\paragraph{The algorithm for general instances} We now sketch our algorithm for general instances, beginning with the following preprocessing rule.}
\lv{
\begin{tlemma}\label{lem:S separates}
Let $(\III,k,S,W_1,W_2)$ be an instance of {\ebdcomp} and let $Z$ be a solution for this instance and let $Z'=Z\cap R_{\cB}[W_1,S]$. Let $\III'$ denote the instance $\III[R_{\cB}[W_1,S]]$. If $S$ disconnects $W_1$ and $W_2$, then $(\III',\vert Z'\vert,S,W_1)$ is a non-separating {\Yes} instance of {\ebdcomp} and conversely for any non-separating solution $Z''$ for the instance $(\III',\vert Z'\vert,S,W_1)$, the set $\hat Z=(Z\setminus Z')\cup Z''$ is a solution for the original instance.

\end{tlemma}

\begin{proof} Suppose that $Z'$ is not a strong backdoor set for the instance $(\III',\vert Z'\vert,S,W_1)$. Let $\cB'$ be the incidence graph of the instance $\III'$. Then, some component of $\cB_{Z'}'$ contains  a set $\C$ of constraints forbidden with respect to $Z'$. However, since $S$ disconnects these constraints from $Z\setminus Z'$, it must be the case that these constraints also occur in a single connected component of $\cB_Z$.
Also, since $Z\setminus Z'$ is disjoint from $\var(\C)$, by Observation \ref{obs:local certificate} we know that $\C$ is also forbidden with respect to $Z'\cup (Z\setminus Z')=Z$, a contradiction. Therefore, $Z'$ is indeed a strong backdoor set for the instance $(\III',\vert Z'\vert,S,W_1)$. Since the connecting gadget has been added on $W_1$ (by assumption), it must be the case that $Z'$ is a non-separating solution for the instance $(\III',\vert Z'\vert,S,W_1)$.

For the converse direction, suppose that the set $\hat Z$ is not a solution for the original instance and let $\C$ be a set of constraints in a connected component of $\cB_{\hat Z}$ which is forbidden with respect to $\hat Z$. Clearly, $\C$ is contained entirely inside one of the sets $R_\cB(W_1,S)$ or $NR_\cB(W_1,S)$. 
We first consider the case when $\C$ is contained in $R_\cB(W_1,S)$. Then, it must be the case that $\C$ is also in a single connected component of $\cB'-Z''$. However, since $Z''$ is assumed to be a strong backdoor set for the CSP instance $\III'$, $\C$ is not forbidden with respect to $Z''$ and hence also not forbidden with respect to $\hat Z$.
On the other hand, we consider the case when $\C$ is contained in $NR_\cB(W_1,S)$. Then $\C$ must be contained in some component of $NR_\cB(W_1,S)-\hat Z$ and in particular in some component of $NR_\cB(W_1,S)-(Z\setminus Z'')=NR_\cB(W_1,S)-(Z\setminus Z')$. Also, since $Z'$ is disjoint from $\var(\C)$, by Observation \ref{obs:local certificate} we know that $\C$ is also forbidden with respect to $Z'\cup (Z\setminus Z')=Z$, a contradiction.
\end{proof}
}

\begin{reduction}\label{red:rule S separates}
Let $(\III,k,S,W_1,W_2)$ be an instance of {\ebdcomp}.
If $S$ disconnects $W_1$ from $W_2$, then compute a non-separating solution $Z'$ for the instance $(\III,k',S,W_1)$ 
where $k'$ is the least possible value of $i\leq k$ such that $(\III,i,S,W_1)$ is a {\Yes} instance. Delete $Z'$ and return the instance $(\III-Z',k-\vert Z'\vert,S,W_2)$.
\end{reduction}

\lv{It follows from Lemma \ref{lem:S separates} that the above rule is correct and we obtain a bound on the running time from that of the algorithm of Lemma \ref{lem:type1 solution}.}
\sv{
 The bound on the running time of the preprocessing rule follows from that of the algorithm of Lemma \ref{lem:type1 solution}.}
Henceforth, we assume that in any given instance of {\ebdcomp}, the above rule is not applicable.
We now move to the description of the subroutine which is at the heart of our main algorithm. \sv{We call a separator $X$ $\ell$-good if $\ell$ is the size of some strong backdoor extending $X\cup S$ in $\III[R_{\cB_S}[W_1,X]\cup S]$. We call an $\ell$-good separator $X$, $\ell$-important if there is no other $\ell$-good separator which dominates $X$.} 
Recall that every solution $Z$ to an instance of {\ebdcomp} by assumption contains an $\ell$-good $W_1$-$W_2$ separator $X$ (here, as well as further on, by separator we implicitly mean a separator in $\cB_S$ unless stated otherwise).

\sv{
 \begin{lemma}\label{lem:full algo} 
   Let $(\III,k,S,W_1,W_2)$ be an instance of {\ebdcomp}, let $0\leq
   \ell,\lambda\leq k$.  There is an algorithm that, given a \emph{valid} tuple
   $<(\III,k,S,W_1,W_2),\lambda,\ell>$ runs in time
   $2^{2^{\bigoh(k)}}\vert \III\vert^{\bigoh(1)}$ and returns a set
   ${\cal R}$ which is disjoint from $S$ and contains at most~$2^{2^{\bigoh(k)}}$ variables such that for
   every $\ell$-important $W_1$-$W_2$ separator $X$ of size at most
   $\lambda$ in $\cB_S$ and for every solution $Z\supseteq X$ for the
   given instance of {\ebdcomp},
 \begin{itemize}
 \item  $\cal R$ intersects $X$  or
 \item  there is a variable in $\cal R$ which is separated from $W$ by $Z$ or 
 \item $\cal R$ intersects $Z\setminus S$.
 \end{itemize}

%
%
 
 \end{lemma}
 }
 
 \lv{
  \begin{lemma}\label{lem:full algo} 
   Let $(\III,k,S,W_1,W_2)$ be an instance of {\ebdcomp}, let $0\leq
   \ell,\lambda\leq k$.  There is an algorithm that, given a \emph{valid} tuple
   $<(\III,k,S,W_1,W_2),\lambda,\ell>$ runs in time
   $2^{2^{\bigoh(k)}}\vert \III\vert^{\bigoh(1)}$ and returns a set
   ${\cal R}$ which is disjoint from $S$ and contains at most~$2^{2^{\bigoh(k)}}$ variables such that for
   every $\ell$-important $W_1$-$W_2$ separator $X$ of size at most
   $\lambda$ in $\cB_S$ and for every solution $Z\supseteq X$ for the
   given instance of {\ebdcomp},
 \begin{itemize}
 \item  $\cal R$ intersects $X$  or
 \item  there is a variable in $\cal R$ which is separated from $W$ by $Z$ or 
 \item $\cal R$ intersects $Z\setminus S$.
 \end{itemize}

%
%
 
 \end{lemma}
 }
 \sv{\begin{proof}[Sketch of Proof]
  Recall that Preprocessing Rule~\ref{red:rule S separates} is assumed to have been exhaustively applied and hence, there must be some $W_1$-$W_2$ path in $\cB_S$. Similarly, if there is no $W_1$-$W_2$ separator of size at most $\lambda$ in $\cB_S$, then we return {\No}, that is, the tuple is invalid.
   Otherwise, we execute the algorithm of Lemma \ref{lem:tight sep computation} to compute a tight $W_1$-$W_2$ separator sequence $\cal I$ of order $\lambda$. We then partition $\cal I$ into $\ell$-good and $\ell$-bad separators. We do this by testing each separator in the sequence for the presence of non-separating solutions (this is sufficient due to the presence of the connecting gadget which is assumed to have been placed on $W_1$), and this can be done in time $2^{\bigoh(\ell^2)}\vert \III\vert^{\bigoh(1)}$ by invoking Lemma \ref{lem:type1 solution}.  Now, let $P_1$ be component-maximal among the $\ell$-good separators in $\cal I$ (if any exist) and let $P_2$ be component-minimal among the $\ell$-bad separators in $\cal I$ (if any exist). We set ${\cal R}:=P_1\cup P_2$. For each $i\in \{1,2\}$, every choice of $P_i^r\subseteq P_i$ and for every boundaried instance $\hat \III$ which needs to be considered for the replacement technique, we construct the glued CSP instance $\III_{P_i^r,\delta}=\III[R[W_1,P_i]\cup S]\otimes_{\delta} \hat \III$. We then recursively invoke this algorithm on the tuple $<(\III_{P_i^r,\delta},k-j,S\cup \tilde S,W_1,P_i\setminus P_i^r),\lambda',\ell'>$ for every $0\leq \lambda'<\lambda$, $1\leq j\leq k-1$ and $0\leq \ell'\leq \ell$, where $\tilde S$ is the annotated set of variables in $\hat \III$. We add the union of  the variable sets returned by these recursive invocations to ${\cal R}$ and return the resulting set. Both the correctness and the running time bound are proved by induction on $\lambda$. Details can be found in the appended full version. 
    \end{proof}
   
   Using Lemma~\ref{lem:full algo}, we prove the following lemma which settles the case for separating instances. Lemma~\ref{lem:ebdcomp algo} then follows.
}

\lv{ 
 \begin{proof}
 Recall that Preprocessing Rule~\ref{red:rule S separates} is assumed to have been exhaustively applied and hence, there must be some $W_1$-$W_2$ path in $\cB_S$. Similarly, if there is no $W_1$-$W_2$ separator of size at most $\lambda$ in $\cB_S$, then we return {\No}, that is, the tuple is invalid.
   Otherwise, we execute the algorithm of Lemma \ref{lem:tight sep computation} to compute a tight $W_1$-$W_2$ separator sequence $\cal I$ of order $\lambda$. We then partition $\cal I$ into $\ell$-good and $\ell$-bad separators. We do this by testing each separator in the sequence for the presence of non-separating solutions (this is sufficient due to the presence of the connecting gadget which is assumed to have been placed on $W_1$), and this can be done in time $2^{\bigoh(\ell^2)}\vert \III\vert^{\bigoh(1)}$ by invoking Lemma \ref{lem:type1 solution}.  Now, let $P_1$ be component-maximal among the $\ell$-good separators in $\cal I$ (if any exist) and let $P_2$ be component-minimal among the $\ell$-bad separators in $\cal I$ (if any exist). We set ${\cal R}:=P_1\cup P_2$. For each $i\in \{1,2\}$, we now do the following.


\begin{figure}[t]
\begin{center}
\includegraphics[height=240 pt, width=320 pt]{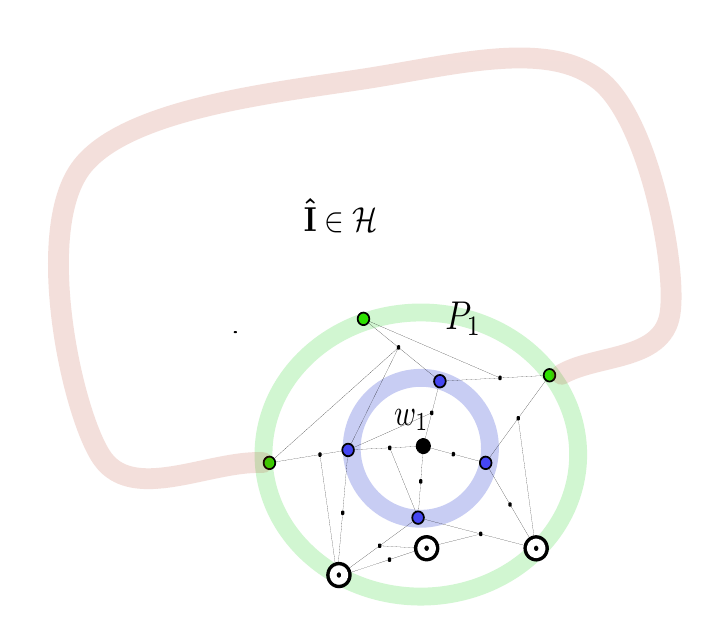}
\caption{An illustration of the instance obtained by gluing some $\hat \III\in \cH$ to the instance $\III[R[W_1,P_1]\cup S]$. }
\label{fig:gluing}
\end{center}
\end{figure}

We execute the algorithm of Lemma \ref{lem:replacement2}, Claim~\ref{cl:five} to compute a family ${\cal H}$ of boundaried CSP instances with an annotated set of variables.  Then, for every choice of $P_i^r\subseteq P_i$, for every instance $\hat \III\in \cH$ with $\vert P_i^r\cup S\vert$ terminals and for every possible bijection $\delta:\partial(\hat \III)\rightarrow P_i^r\cup S$, we construct the glued CSP instance $\III_{P_i^r,\delta}=\III[R[W_1,P_i]\cup S]\otimes_{\delta} \hat \III$ (see Figure \ref{fig:gluing}, where the boundary of the instance $\III[R[W_1,P_i]\cup S]$ is defined as $P_i^r\cup S$ and the connecting gadget is added on $W_1\cup P_i^r$. We then recursively invoke this algorithm on the tuple $<(\III_{P_i^r,\delta},k-j,S\cup \tilde S,W_1,P_i\setminus P_i^r),\lambda',\ell'>$ for every $0\leq \lambda'<\lambda$, $1\leq j\leq k-1$ and $0\leq \ell'\leq \ell$, where $\tilde S$ is the annotated set of variables in $\hat \III$. We add the union of  the variable sets returned by these recursive invocations to ${\cal R}$ and return the resulting set.
This completes the description of the algorithm. We now proceed to the proof of correctness of this algorithm.

\paragraph{Correctness.} We prove the correctness by induction on $\lambda$. Consider the base case, when $\lambda=1$ and there is a path from $W_1$ to $W_2$ in $\cB_S$. We argue the correctness of the base case as follows. 
Since $X$ has size 1, it cannot be incomparable with any distinct $W_1$-$W_2$ separator of size 1. Therefore, $X$ has to be equal to $P_1$ or covered by $P_1$ ($P_1$ exists since $X$ itself is $\ell$-good by assumption).
 In either case we are correct. The former case is trivially accounted for since $P_1$ is contained in $\cR$ and the latter case is accounted for because $P_1$ clearly well-dominates $X$. We now move to the induction step with the induction hypothesis being that the algorithm runs correctly (the output satisfies the properties claimed in the statement of the lemma) on all tuples where $\lambda<\hat \lambda$ for some $\hat \lambda\geq 2$. Now, consider an invocation of the algorithm on a tuple with $\lambda=\hat \lambda$.

Let $Z$ be a solution for this instance containing the $\ell$-important separator $X$, i.e., $X\subseteq (Z\setminus S)$, of size at most $\lambda$. If $X$ intersects $P_1\cup P_2$, then the algorithm is correct since $\cal R$ intersects $Z\setminus S$. Therefore, we may assume that $X$ is disjoint from $P_1\cup P_2$. 

Now, suppose that $X$ is covered by $P_1$. In this case, since $P_1$ is also $\ell$-good and has size at most $\lambda$, by the definition of well-domination, $P_1$ well-dominates $X$, contradicting our assumption that $X$ is $\ell$-important. 

Similarly, by the Monotonocity Lemma (Lemma \ref{lem:monotone}), since $X$ is $\ell$-good and $P_2$ is not, it cannot be the case that $X$ covers $P_2$. Now, suppose that $X$ covers $P_1$ and is itself covered by $P_2$. However, due to the maximality of the tight separator sequence, $X$ must be contained in $\cal I$. But this contradicts our assumption that $P_1$ is a component-maximal $\ell$-good separator in the sequence $\cal I$ and $P_1\neq X$. 

\begin{figure}[t]
\begin{center}
\includegraphics[height=240 pt, width=270 pt]{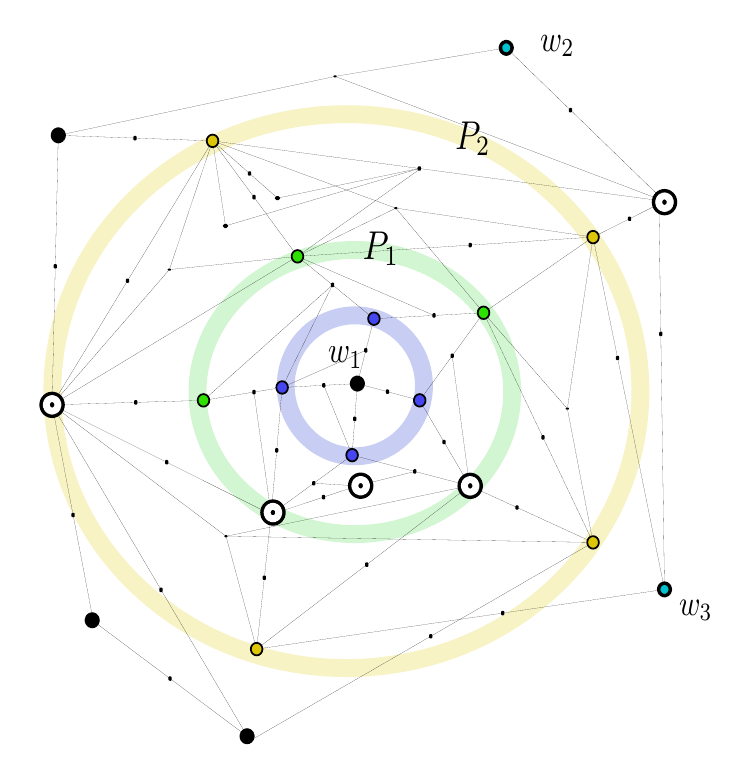}\caption{An illustration of the case where $X$ (the dotted circles) is incomparable with $P_1$ and $P_2$. }
\label{fig:incomparable}
\end{center}
\end{figure}

Finally, we are left with the case when $X$ is incomparable with $P_1$ (if $P_1$ is defined) or $P_2$ (otherwise). Without loss of generality, suppose that $X$ is incomparable with $P_1$ (see Figure \ref{fig:incomparable}). The argument for the case when $P_1$ does not exist is analogous and follows by simply replacing $P_1$ with $P_2$ in the proof.

Let $K\subseteq Z$ be a strong backdoor set for $\III[R[W_1,X]\cup S]$ extending $X\cup S$. If $P_1\cap K$ is non-empty, then $P_1\cap Z$ is non-empty as well, and since ${\cal R}$ contains the vertices in $P_1$, the algorithm is correct. Therefore, we assume that $P_1$ and $K$ are disjoint. Now, let $X^r=R_{\cB_S}(W_1,P_1)\cap X$ and $X^{nr}=X\setminus X^r$. Similarly, let $P_1^r=R_{\cB_S}(W_1,X)\cap P_1$ and $P_1^{nr}=P_1\setminus P_1^r$. Since $X$ and $P_1$ are incomparable, the sets $X^{r}$, $X^{nr}$, $P_1^r$ and $P_1^{nr}$ must all be non-empty. 
  Let $K^r=(K\cap R(W_1,P))\cup S$. Furthermore, if any variable in $P_1^r$ is not in the same component as $W_1$ in $\cB_Z$, then ${\cal R}$ contains a variable separated from $W$ by $Z$, also implying that the algorithm is correct. Hence, we may assume that $P_1^r$ is contained in the same component as $W_1$ in $\cB_Z$. Observe that the sets defined above satisfy the premises of Lemma \ref{lem:replacement2} with $P=P_1$ and $Q=X$ in the statement of the lemma. Therefore, there is a $\vert P_1^r\cup S\vert$-boundaried instance $\hat \III$ with an annotated set $\tilde S$ and an appropriate bijection $\mu:\partial(\hat \III)\rightarrow P_1^r\cup S$ with the properties claimed in the statement of Lemma \ref{lem:replacement2}. Now, consider the recursion of the algorithm on the tuple $<(\III_{P_1^r,\mu},k_1,S\cup \tilde S,W_1,P_1^{nr}),\lambda',\ell'>$, where $\III_{P_1^r,\mu}$ is the instance obtained by gluing together $\III[R[W_1,P_1]\cup S]$ (with $P_1^r\cup S$ as the boundary) and $\hat \III$ via the bijection $\mu$ with the connecting gadget added on $W_1\cup P_1^r$, $\lambda'=\vert X^r \vert$, $k_1=\vert K^r\vert$, and $\ell'=\vert K^r\vert$.

In order to apply the induction hypothesis on the execution of the algorithm on this tuple, we need to prove that the tuple is `valid', that is, it satisfies the conditions in the premise of the lemma. In order to prove this, it is sufficient for us to prove that $(\III_{P_1^r,\mu},k_1,S\cup \tilde S,W_1,P_1^{nr})$ is indeed a valid instance of {\ebdcomp}. For this to hold, it must be the case that $W_1\cup P_1^{nr}\cup S$ is a strong-backdoor set for the CSP instance $\III_{P_1^r,\mu}$. But this property is indeed guaranteed by Lemma \ref{lem:replacement2}, Claim~\ref{cl:one}. Therefore, the tuple satisfies the conditions in the premise of the lemma and since $\lambda'<\lambda=\hat \lambda$, we may apply the induction hypothesis.

Since $X$ is $\ell$-important, from Lemma \ref{lem:replacement2}, Claims~\ref{cl:two} and~\ref{cl:three} it follows that $X^r$ must also be $k_1$-important in $\III_{P_1^r,\mu}$. By the induction hypothesis, the algorithm is correct on this tuple and the returned set, call it $\cR'$, either intersects $X^r$ (due to the aforementioned argument using Claims~\ref{cl:two} and~\ref{cl:three}), or contains a variable separated from $W_1\cup P^{nr}$ by $K^r$ or contains a variable in $K^r$. 
 By Lemma \ref{lem:replacement2}, Claim~\ref{cl:four} it holds that in the second case $\cR'$ contains a variable separated from $W$ by $Z$. In the third case, $\cR'$  contains a variable in $Z\setminus S$ because $K^r\subseteq K\subseteq Z$ and $\cR'$ is disjoint from $S\cup \tilde S\supseteq S$. Since $\cR'\subseteq \cR$, we conclude the correctness of the algorithm and
we now move on to the analysis of the running time and the size of the returned set $\cR$.

\paragraph{Bounding the set $\cR$.} 
Recall that since $\lambda$, which is bounded above by $k$, is required to be non-negative in a valid tuple, the depth of the search tree is bounded by $k$. Furthermore, the number of branches initiated at each node of the search tree is at most $k^3\cdot k!\cdot g(k)$, where $g(k)$ is the number of boundaried CSP instances in the set $\cH$ ($k^3$ for the choice of $\lambda'$, $j$ and $\ell'$ and $k!$ for the choice of the bijection $\delta$). Since Lemma \ref{lem:replacement2} guarantees a bound of $2^{2^{\eta k}}$ on $\vert \cH\vert$ for some constant $\eta$, we conclude that the number of internal nodes in the search tree is bounded by $2^{2^{\eta' k}}$ for some constant $\eta'$. Finally, since at each internal node we add at most $2k$ vertices (corresponding to $P_1\cup P_2$), we conclude that the set returned has size $2^{2^{\bigoh(k)}}$.

%
%
%

\paragraph{Running Time.} The analysis for the running time is similar to the proof of the bound on $\cR$. We already have a bound on the number of nodes of the search tree. The claimed bound on the running time of the whole algorithm follows from the observation that the time spent at each node of the search tree is dominated by the time required to execute the algorithms of Lemma \ref{lem:type1 solution} and Lemma \ref{lem:replacement2}, which in turn is bounded by $2^{2^{\bigoh(k)}}\vert \III\vert^{\bigoh(1)}$.
This completes the proof of the lemma.
\end{proof}
}


%
%
%
%

\begin{lemma}\label{lem:type2 solution}
 There is an algorithm that, given an instance $(\III,k,S,W_1,W_2)$ of {\ebdcomp}, runs in time $2^{2^{\bigoh(k)}}\vert \III\vert^{\bigoh(1)}$ and either computes a solution for this instance which is a $W_1$-$W_2$ separator or correctly concludes that no such solution exists.
\end{lemma}

\lv{
\begin{proof}
For every $0\leq \ell,\lambda\leq k$, we invoke Lemma \ref{lem:full algo} on the tuple $<(\III,k,S,W_1,W_2),\lambda,\ell>$ to compute a set $\cR_{\ell,\lambda}$. We then set $\cal R$ to be the set obtained by taking the union of the sets $\cR_{\ell,\lambda}$ for all possible values of $\ell$ and $\lambda$. Following that, we simply branch on which vertex $v$ in $\cR$ is added to $S$, creating a new instance $(\III,k,S\cup \{v\},W_1,W_2)$ of {\ebdcomp}. If $|S\cup \{v\}|>k$ we return \No. If Preprocessing Rule~\ref{red:rule S separates} applies, we use it to reduce the instance; if this results in a non-separating instance, we use Lemma~\ref{lem:type1 solution} to solve the instance. Otherwise, we iterate all of the above on this new instance. 

The bound on the running time of this algorithm follows from the total depth of the branching tree being bounded by $k$ and the width of each branch being bounded by $|\cR|\leq 2^{2^{\bigoh(k)}}$ (the cost of branching over $\ell, \lambda$ and of applying Lemma~\ref{lem:type1 solution} and Preprocessing Rule~\ref{red:rule S separates} is dominated by the above). The correctness follows from the correctness of Lemma~\ref{lem:full algo}, of Preprocessing Rule~\ref{red:rule S separates} and of Lemma~\ref{lem:type1 solution}. This completes the proof of the lemma and we now have our algorithm to handle separating instances of {\ebdcomp}.
\end{proof}

We conclude this section by combining the algorithms for separating and non-separating instances to present our complete algorithm for {\ebdcomp} (Lemma \ref{lem:ebdcomp algo}). 

\begin{proof}[Proof of Lemma~\ref{lem:ebdcomp algo}]
Let $(\III,k,S,W)$ be the input instance of {\ebdcomp}. We first apply Lemma \ref{lem:type1 solution} to check if there is a non-separating solution for this instance. If not, then we branch over all $W_1\subset W$ and for each such choice of $W_1$ we add the connecting gadget on $W_1$ and apply Lemma \ref{lem:type2 solution} to check if $(\III,k,S,W_1,W_2=W\setminus W_1 )$ has solution which is a $W_1$-$W_2$ separator.
The correctness and claimed running time bound both follow from those of Lemma \ref{lem:type1 solution} and Lemma \ref{lem:type2 solution}.
%
\end{proof} 
}

\section{Concluding Remarks}
We have presented an FPT algorithm that can find strong backdoors to
scattered base classes of CSP and \sharpCSP. This algorithm allows us
to lift known tractability results based on constraint languages from
instances over a single tractable language to instances containing a
mix of constraints from distinct tractable languages. The instances
may also contain constraints that only belong to a tractable language
after the backdoor variables have been instantiated, where different
instantiations may lead to different tractable languages.
Formally we have applied the algorithm to CSP and \#CSP, but it
clearly applies also to other versions of CSP, such as MAX-CSP
(where the task is to simultaneously satisfy a maximum number of
constraints) or various forms of weighted or valued CSPs. 

Our work opens up several avenues for future research. First of all, the runtime bounds for finding backdoors to scattered base classes provided in this work are very likely sub-optimal due to us having to obtain a unified algorithm for \emph{every} scattered set of finite constraint languages. Therefore, it is quite likely that a refined study for scattered classes of specific constraint languages using their inherent properties will yield significantly better runtimes. Secondly, graph modification problems and in particular the study of efficiently computable modulators to various graph classes has been an integral part of parameterized complexity and has led to the development of several powerful tools and techniques. We believe that the study of modulators to `scattered graph classes' could prove equally fruitful and, as our techniques are mostly graph based, our results as well as techniques could provide a useful starting point towards future research in this direction.

\sv{\newpage}
\bibliographystyle{abbrv}
\bibliography{references}

\end{document}